\documentclass[10pt]{article}

\usepackage[a4paper,height=190mm,width=142mm]{geometry}

\usepackage[utf8]{inputenc}
\usepackage{amsthm}
\usepackage{amssymb}
\usepackage{latexsym}
\usepackage{amsmath}
\usepackage{amsfonts}
\usepackage{mathrsfs}
\usepackage{amstext}
\usepackage{enumitem}
\setitemize{noitemsep,topsep=2pt,parsep=2pt,partopsep=2pt}
\usepackage{mathtools}

\DeclarePairedDelimiter{\floor}{\lfloor}{\rfloor}
\mathtoolsset{showonlyrefs}
\usepackage{stmaryrd}
\usepackage{tocloft}
\usepackage{cite}
\usepackage{graphicx}
\usepackage{xcolor}
\usepackage{color,hyperref}
\definecolor{darkblue}{rgb}{0.0,0.0,0.5}
\hypersetup{
  linkcolor  = darkblue,
  citecolor  = darkblue,
  urlcolor   = darkblue,
  colorlinks = true,
}

\newcommand{\mfh}{\mathfrak{h}}
\newcommand{\Om}{\Omega}
\newcommand{\mcF}{\mathcal{F}}
\newcommand{\hil}{\mathcal{H}} 
\newcommand{\lan}{\langle}
\newcommand{\ran}{\rangle}

\newcommand{\bN}{\mathbb{N}} 
\newcommand{\bE}{\mathbb{E}}
\newcommand{\bC}{\mathbb{C}}

\newcommand{\bR}{\mathbb{R}}

\newcommand{\bS}{\mathbb{S}}
\newcommand{\cS}{\mathcal{S}}

\newcommand{\cR}{\mathcal{R}}
\newcommand{\cT}{\mathcal{T}}
\newcommand{\cF}{\mathcal{F}}
\newcommand{\sS}{\mathscr{S}}
\newcommand{\sD}{\mathscr{D}}

\newcommand{\rx}{\mathrm{x}}
\newcommand{\ry}{\mathrm{y}}

\newcommand{\rd}{\mathrm{d}}
\newcommand{\re}{\mathrm{e}}

\newcommand{\rc}{\mathrm{c}}

\newcommand{\supp}{\mathrm{supp}}
\makeatletter
\newcommand*\botimes{{\mathpalette\botimes@{1.5}}}
\newcommand*\botimes@[2]{\mathbin{\vcenter{\hbox{\scalebox{#2}{\hspace{0.0mm}$\m@th#1\otimes$\hspace{0.0mm}}}}}}
\newcommand*\Cdot{{\mathpalette\Cdot@{.6}}}
\newcommand*\Cdot@[2]{\mathbin{\vcenter{\hbox{\scalebox{#2}{\hspace{0.5mm}$\m@th#1\bullet$\hspace{0.5mm}}}}}}
\makeatother

\newtheorem{thm}{Theorem}
\newtheorem{lem}[thm]{Lemma}
\newtheorem{prop}[thm]{Proposition}

\theoremstyle{remark}
\newtheorem{rem}[thm]{Remark}
\theoremstyle{remark}

\theoremstyle{definition}
\newtheorem{dfn}[thm]{Definition}

\numberwithin{equation}{section}
\numberwithin{thm}{section}
\numberwithin{ass}{section}

\begin{document}

\title{Stochastic quantization of two-dimensional \\$P(\Phi)$ Quantum Field Theory}
\author{Pawe{\l} Duch, Wojciech Dybalski and Azam Jahandideh
\\
{\normalsize Faculty of Mathematics and Computer Science} \\  
{\normalsize Adam Mickiewicz University in Pozna\'n}\\
{\normalsize ul. Uniwersytetu Pozna\'nskiego 4, 61-614 Pozna\'n, Poland}\\
{\normalsize\{pawel.duch, wojciech.dybalski, azajah\}@amu.edu.pl}
}
\date{\today}

\maketitle

\begin{abstract}
We give a simple and self-contained construction of of the $P(\Phi)$ Euclidean Quantum Field Theory in the plane and verify the Osterwalder-Schrader axioms: translational and rotational invariance, reflection positivity and regularity. In the intermediate steps of the construction we study measures on spheres. In order to control the infinite volume limit we use the parabolic stochastic quantization equation and the energy method. To prove the translational and rotational invariance of the limit measure we take advantage of the fact that the symmetry groups of the plane and the sphere have the same dimension.
\end{abstract}

\tableofcontents

\section{Introduction}

We revisit the construction of the Euclidean two-dimensional $P(\Phi)$ quantum field theory model also known as the $P(\Phi)_2$ model. The main new contribution is a simple construction of the infinite volume measure of this model using the stochastic quantization technique~\cite{parisi1981} and the verification of all Osterwalder-Schrader axioms with the exception of clustering. By the Osterwalder-Schrader reconstruction theorem~\cite{OS2} this yields the existence of the theory in the Lorentzian signature that satisfies all the Wightman axioms possibly with the exception of the uniqueness of the vacuum. Let us point out that the proof of the invariance of the infinite volume measure under all of the Euclidean transformations of the plane, which is one of the Osterwalder-Schrader axioms, is quite non-trivial. In fact, finite-volume measures, which are typically introduced in the intermediate steps of the construction, are usually invariant only under a certain subset of the Euclidean transformations. The novelty of the approach taken in the present work is to study finite-volume $P(\Phi)_2$ measures defined on spheres in the intermediate steps of the construction. Such measures are invariant under the action of the three-dimensional Lie group of the rotations of the sphere (in contrast, measures defined on a torus are only invariant under the action of the two-dimensional Lie group). To prove the Euclidean invariance of the infinite volume measure we crucially use the fact that the symmetry groups of the plane and the sphere have the same dimension.

Fix $n\in2\bN_+$, $n\geq 4$, and a real polynomial
\begin{equation}
 P(\tau) = \sum_{m=0}^n a_m \tau^m,
 \qquad \tau\in\bR,\qquad a_0,\ldots,a_{n-1}\in\bR, \quad a_n=1/n\,.
\end{equation}
Let $\bS_R$ be a  round two-dimensional sphere of radius $R\in\bN_+$ with the metric induced from~$\bR^3$. The Laplace-Beltrami operator on $\bS_R$ is denoted by $\Delta_R$ and the canonical Riemannian volume form on $\bS_R$ is denoted by $\rho_R$. For $R\in\bN_+$ a probability measure $\mu_R$ on $\sD'(\bS_R)$ is defined by
\begin{equation}\label{eq:measure_mu_R}
 \mu_R(\rd\phi):=\frac{1}{\mathcal{Z}_R}\exp\left(-\int_{\bS_R} \lambda :\!P(\phi(\rx))\!:\,\rho_R(\rd \rx)\right) \nu_R(\rd\phi),
\end{equation}
where $\lambda\in(0,\infty)$ is the coupling constant, $\mathcal{Z}_R\in(0,\infty)$ is the normalization factor, $\nu_R$ is the Gaussian measure on $\sD'(\bS_R)$ with covariance $G_R:=(1-\Delta_R)^{-1}$ and $:\Cdot:$ denotes the Wick ordering. The measure $\mu_R$ is called the $P(\Phi)_2$ measure on $\bS_R$. In Sec.~\ref{sec:uv_limit} we review the construction of this measure based on the Nelson hypercontractivity argument~\cite{Nelson}. By construction, $\mu_R$ is invariant under rotations of~$\bS_R$. Let $\jmath_R\,:\,\bR^2\to\bS_R$ be the parametrization of $\bS_R$ by the stereographic coordinates. By $\jmath_R^*\sharp\mu_R$ we denote the measure on $\sS'(\bR^2)$ obtained by the push-forward of $\mu_R$ by the pullback $\jmath_R^*\,:\,\sD'(\bS_R)\to\sS'(\bR^2)$. The main result of the paper is the following theorem.

\begin{thm}
The sequence of measures $(\jmath_R^*\sharp\mu_R)_{R\in\bN_+}$ on $\sS'(\bR^2)$ has a weakly convergent subsequence. Every accumulation point $\mu$ of $(\jmath_R^*\sharp\mu_R)_{R\in\bN_+}$ is invariant under the Euclidean symmetries of the plane and reflection positive. Moreover, there exists a ball $B\subset\sS(\bR^2)$ with respect to some Schwartz semi-norm centered at the origin such that for all $f\in B$ it holds 
\begin{equation}\label{eq:main_thm_bound}
 \int\exp(\phi(f)^n)\,\mu(\rd\phi) 
 \leq 2.
\end{equation}
\end{thm}
\begin{rem}
Any accumulation point $\mu$ of $(\jmath_R^*\sharp\mu_R)_{R\in\bN_+}$ is called the $P(\Phi)_2$ measure on the plane.
\end{rem}
\begin{rem}
The bound~\eqref{eq:main_thm_bound} implies that $\mu$ is non-Gaussian as Gaussian measures do not integrate functions growing so fast. Moreover, the Osterwalder-Schrader regularity axiom~\cite{OS2} is an immediate consequence of this bound.
\end{rem}
\begin{rem}
By the above theorem $\mu$ satisfies all the Osterwalder-Schrader axioms~\cite{OS2} possibly with the exception of the cluster property (the  decay of correlation functions). It is known that the $P(\Phi)_2$ measure on the plane is unique provided $\lambda\in(0,\infty)$ is sufficiently small~\cite{GJS74}. In general uniqueness does not hold and the model exhibits phase transitions \cite[Ch.~16]{GJ85}. The cluster property is only expected to hold in pure phases. Our construction of the $P(\Phi)_2$ measure does not need any smallness assumption on $\lambda$. However, it does not give any information about the uniqueness of the infinite volume limit or the decay of correlation functions. In what follows we set $\lambda=1$.
\end{rem}

\begin{proof}
The existence of a weakly convergent subsequence of $(\jmath_R^*\sharp\mu_R)_{R\in\bN_+}$ follows from tightness and Prokhorov's theorem. The proof of tightness is presented in Sec.~\ref{sec:tightness} and uses parabolic stochastic quantization combined with a PDE energy estimate. The invariance of $\mu$ under the Euclidean symmetries is established in Sec.~\ref{sec:Euclidean} and is based on the fact that for all $R\in\bN_+$ the measure $\mu_R$ is invariant under the group of rotations of $\bS_R$. The proof that $\mu$ is reflection positive is given in Sec.~\ref{sec:reflection} and is based on the fact that for all $R\in\bN_+$ the measure $\mu_R$ is reflection positive. The bound~\eqref{eq:main_thm_bound} is proved in Sec.~\ref{sec:integrability} with the use of the Hairer-Steele argument~\cite{HairerSteele}. 
\end{proof}

The $P(\Phi)_2$ model has been extensively studied in the literature and is arguably the simplest example of an interacting QFT. The overview of various approaches used to construct this model can be found in the books~\cite{GlimmJaffe,Simon} and the review article~\cite{Summers}. Since the finite-volume $P(\Phi)_2$ measure is absolutely continuous with respect to the free field measure the construction of the model in finite volume is quite elementary and was given by Nelson in~\cite{Nelson,DiGenova}. We also mention constructions of the  $P(\Phi)_2$ models on  de Sitter spacetime, whose Euclidean counterpart is a sphere~\cite{FHN75, BJM23,JM18}. The construction of the infinite volume $P(\Phi)_2$ model directly in the Lorentzian signature including the verification of the Haag-Kastler axioms was carried out in the early 70's by Glimm and Jaffe~\cite{GJ85}. The construction was later revisited in the Euclidean setting.  For $\lambda>0$ sufficiently small a complete construction of the Euclidean $P(\Phi)_2$ model and the verification of all of the Osterwalder-Schrader axioms including exponential decay of correlations was given in \cite{GJS74} (see also~\cite{GlimmJaffe,GJ85}) using the cluster expansion technique. Let us also mention an alternative technique based on the correlation equalities that works for all $\lambda$ positive and polynomials $P(\tau) = Q(\tau) - h\,\tau $ such that $Q$ is an even polynomial and $h\in\bR$, which was originally developed in~\cite{GRS75} (see also~\cite{Simon,GlimmJaffe}). We stress that the method of our paper  works for all $\lambda$ positive and all polynomials $P$ bounded from below. In order to control the infinite volume limit we have to prove certain bounds for moments of the regularized measures uniform in both the ultraviolet and infrared cutoffs. To this end, we use the parabolic stochastic quantization and the energy method. Let us note that a similar approach has already been used, for example, in ~\cite{AlbeverioKusuoka,gubinelli2018pde,AlbeverioKusuoka2} to construct the $\Phi^4_3$ model. The analysis of the above-mentioned references can be trivially adapted to the case of the much simpler $P(\Phi)_2$ model. Let us point out that in~\cite{gubinelli2018pde} the infinite volume measure is constructed as a limit of a sequence of measures defined on tori of increasing size. The symmetry group of the torus consists of translations, reflections and rotations by a multiple of $\pi/2$ and it is easy to prove that the infinite volume measure also has these symmetries. However, it is not clear whether it is invariant under all rotations. In the construction of~\cite{AlbeverioKusuoka2} an infrared cutoff preserving the rotational invariance was used. The rotational invariance of the infinite volume limit is then obvious. However, the translational invariance is far from clear as it is explicitly broken by the infrared cutoff.
In~\cite{AlbeverioKusuoka} infinite volume limit was not investigated. 
In the present work we study $P(\Phi)_2$ measures defined on spheres of increasing radius. In order to show the invariance of the infinite volume $P(\Phi)_2$ measure under all Euclidean transformations we take advantage of the fact that the symmetry groups of the plane and the sphere have the same dimension. We remark that using the strategy of this paper it should also be possible to construct the infinite volume $\Phi^4_3$ measure and prove its invariance under all Euclidean transformations of $\bR^3$ by appropriately adapting the analysis of~\cite{gubinelli2018pde,AlbeverioKusuoka2}.

The plan of the paper is as follows. In Sec.~\ref{sec:uv_limit} we introduce the $P(\Phi)_2$ measure $\mu_{R,N}$ on the sphere $\bS_R$ with a certain UV cutoff $N\in\bN_+$ in the frequency space and prove the convergence of $\mu_{R,N}$ as $N\to\infty$ to the measure $\mu_R$ formally given by Eq.~\eqref{eq:measure_mu_R}. We also investigate a certain auxiliary measure $\mu_{R,N}^g$, which coincides with $\mu_{R,N}$ when $g=0$. The auxiliary measure $\mu_{R,N}^g$ is used in Sec.~\ref{sec:integrability} to prove the bound~\eqref{eq:main_thm_bound}. In Sec.~\ref{sec:stochastic_quantization} we study the stochastic quantization equation of the measure $\mu_{R,N}^g$. We also introduce a related stochastic PDE obtained with the use of the so-called Da Prato-Debussche trick that, in contrast to the former SPDE, is well defined in the limit $N\to\infty$. In Sec.~\ref{sec:a_priori_bound} we apply the energy technique to prove a certain a priori bound for the latter SPDE. The a priori bound is uniform in both the radius of the sphere $R$ as well as the UV cutoff $N$ and is the main ingredient in the proof of the existence of the infinite volume limit of the measures $\mu_R$, which is presented in Sec.~\ref{sec:tightness}. In order to make sense of the infinite volume limit we have to first identify the measure $\mu_R$ on $\sD'(\bS_R)$ with a certain measure on $\sS'(\bR^2)$. To this end, we use the stereographic projection of the sphere $\bS_R$ onto the plane $\bR^2$ whose definition is recalled in Sec.~\ref{sec:stereographic}. Sec.~\ref{sec:reflection} is devoted to the proof of the reflection positivity. In Sec.~\ref{sec:Euclidean} we show that an infinite volume $P(\Phi)_2$ measure is invariant under translations and rotations. The proof relies on the invariance of the finite volume measure $\mu_R$ on $\sD'(\bS_R)$ under all rotations of the sphere $\bS_R$ and some elementary properties of the stereographic projection. More specifically, we use the observation that if the radius $R$ of the sphere is very big, then the Euclidean transformations of the plane are well approximated by appropriately chosen rotations of the sphere. In Appendix~\ref{sec:spaces} we recall the definitions and collect useful facts about various function spaces used in the paper. Appendix~\ref{sec:preliminaries} contains some auxiliary results. In Appendix~\ref{sec:stochastic_estimates} we prove uniform bounds for moments of norms of Wick polynomials of regularized free fields. 

\section{Ultraviolet limit}\label{sec:uv_limit}

In this section we recall the construction of the $P(\Phi)_2$ measure on $\bS_R$ based on the Nelson hypercontractivity estimate~\cite{Nelson}. We first introduce the measures $(\mu_{R,N})_{N\in\bN_+}$ with the UV regularization and show that the limit $\lim_{N\to\infty}\mu_{R,N}=\mu_R$ exists in the sense of weak convergence of measures. For $R,N\in\bN_+$ we define the bounded operators $G_R,K_{R,N}\,:\,L_2(\mathbb{S}_R)\to L_2(\mathbb{S}_R)$,
\begin{equation}
 G_R:=(1-\Delta_R)^{-1},\qquad K_{R,N}:=(1-\Delta_R/N^2)^{-1} 
\end{equation}
and a probability measure on $\sD'(\bS_R)$
\begin{equation}\label{eq:regularized_measure}
 \mu_{R,N}(\rd\phi):=\frac{1}{\mathcal{Z}_{R,N}}\exp\left(-\int_{\bS_R} P(\phi(\rx),c_{R,N})\,\rho_R(\rd \rx)\right) \nu_{R,N}(\rd\phi),
\end{equation}
where $\nu_{R,N}$ is the Gaussian measure on $\sD'(\bS_R)$ with covariance $G_{R,N}:=K_{R,N}G_RK_{R,N}$,
\begin{equation}\label{eq:counterterm}
c_{R,N}:=\int_{\sD'(\bS_R)}\phi(\rx)^2\,\nu_{R,N}(\rd\phi) = \mathrm{Tr}(K_{R,N}G_RK_{R,N})/4\pi R^2
\end{equation}
is the so-called counterterm and
\begin{equation}
 P(\tau,c):= \sum_{m=0}^n a_m\,\sum_{k=0}^{\floor{m/2}} \frac{(-1)^k m!}{(m-2k)!k!2^k} c^k \tau^{m-2k}, 
 \qquad \tau,c\in\bR.
\end{equation}
Note that by Lemma~\ref{lem:trace} there exists $C\in(0,\infty)$ such that for all $N,R\in\bN_+$ it holds
\begin{equation}\label{eq:bound_counterterm}
 |c_{R,N}-1/2\pi\,\log N|\leq C.
\end{equation}
Observe also that $P(\phi(\rx),c_{R,N})$ is obtained by Wick-ordering $P(\phi(\rx))$ with respect to the regularized measure $\nu_{R,N}$.
Accordingly, the sum over $k$  in the definition of $P(\tau,c)$ amounts for  $c\geq 0$ to $\tau \mapsto c^{m/2}H_m(\tau/c^{1/2})$, where $H_m$, $m\in \bN_0$, are the Hermite polynomials, cf. Appendix~\ref{sec:stochastic_estimates}.   

Actually, in order to  establish the bound~\eqref{eq:main_thm_bound} we will study a more general class of probability measures
\begin{equation}\label{eq:regularized_measure_g}
 \mu_{R,N}^g(\rd\phi):=\frac{1}{\mathcal{Z}^g_{R,N}}\exp\left(\phi(g)^n/n\right) \mu_{R,N}(\rd\phi),
\end{equation}
with $g\in C^\infty(\bS_R)$ such that 
\begin{equation}\label{eq:g_conditions}
\|g\|^n_{L_{n/(n-1)}(\bS_R)}\leq 1/2,
\qquad
\|\Delta_R g\|^n_{L_{n/(n-1)}(\bS_R)}\leq 1/2.
\end{equation}
The usefulness of the measure $\mu_{R,N}^g$ comes from Lemma~\ref{lem:mu_f_bound}, which says that in order to show the bound~\eqref{eq:main_thm_bound} it is sufficient to prove a certain uniform bound for some finite moment of the measure $\mu_{R,N}^g$. In Lemma~\ref{lem:mu_measure_well_defined} we show that the measure $\mu_{R,N}^g$ is well defined. Proposition~\ref{prop:uv_limit} implies in particular that for every $R\in\bN_+$ the sequence of measures $(\mu_{R,N}^g)_{N\in\bN_+}$ converges weakly to a measure denoted by $\mu_R^g$. If $g=0$, then $\mu_R^g$ coincides with the $P(\Phi)_2$ measure on $\bS_R$, which is denoted by $\mu_R$. Moreover, the measures $\mu_R^g$ and $\mu_R$ are related by a formula analogous to~\eqref{eq:regularized_measure_g}.

\begin{lem}\label{lem:mu_measure_well_defined}
For all $R,N\in\bN_+$ and $g\in C^\infty(\bS_R)$ such that $\|g\|^n_{L_{n/(n-1)}(\bS_R)}\leq 1/2$ the measure $\mu_{R,N}^g$ is well-defined and both $\nu_{R,N}$ and $\mu_{R,N}^g$ are concentrated on $L^1_2(\bS_R)\subset\sD'(\bS_R)$.
\end{lem}
\begin{rem}
We identify implicitly a function $\phi\in L_1(\bS_R)$ with a distribution $\varphi\in\sD'(\bS_R)$ defined by $\varphi(f)\equiv\langle\varphi,f\rangle:=\int_{\bS_R} \phi(\rx)f(\rx)\,\rho_R(\rd\rx)$.
\end{rem}

\begin{proof}
By Lemma~\ref{lem:stochastic_X_not_uniform} the measure $\nu_{R,N}$ is concentrated on $L^1_2(\bS_R)$. By the Sobolev embedding $L^1_2(\bS_R)\subset L_n(\bS_R)$ stated in Lemma~\ref{lem:embedding_sphere}, the bound~$\phi(g)^n/n\leq \|\phi\|^n_{L_n(\bS_R)}/2n$ and the boundedness from below of the polynomial $\tau\mapsto P(\tau,c_{R,N})-\tau^n/2n$ the function
\begin{equation}
 \mathcal{U}^g_{R,N}\,:\,
 L_n(\bS_R)\ni \phi \mapsto \exp\left(\frac{1}{n}\phi(g)^n-\int_{\bS_R} P(\phi(\rx),c_{R,N})\,\rho_R(\rd \rx)\right)  \in (0, \infty)
\end{equation}
is bounded and continuous. Moreover, $\mathcal{Z}_{R,N}\mathcal{Z}_{R,N}^g = \|\mathcal{U}^g_{R,N}\|_{L_1(\sD'(\bS_R),\nu_{R,N})}\geq 1$ by the Jensen inequality and Lemma~\ref{Nualart-Lemma-1-1}. This proves the claims about the measure $\mu_{R,N}^g$.
\end{proof}

\begin{dfn}\label{dfn:X_Y}
We define $X_R$ to be the Gaussian random variable valued in $\sD'(\bS_R)$ with mean zero and covariance $G_R$. We set $X_{R,N}:=K_{R,N}X_R$,
\begin{equation}
\begin{gathered}
 X_{R,N}^{:m:}(\rx):=\sum_{k=0}^{\floor{m/2}} \frac{(-1)^k m!}{(m-2k)!k!2^k} (c_{R,N})^k X_{R,N}^{m-2k}(\rx),
 \quad
 X_{R,N}^{:m:}(h):=\int_{\bS_R} X_{R,N}^{:m:}(\rx)\,h(\rx)\,\rho_R(\rd \rx),
 \\
 Y_{R,N}:=\sum_{m=0}^n a_m X_{R,N}^{:m:}(1_{\bS_R}),
 \qquad
 Y_{R,N}^g:=Y_{R,N}-X_{R,N}(g)^n/n,
\end{gathered} 
\end{equation}
where $h\in L_\infty(\bS_R)$ and $1_B$ denotes the characteristic function of the set $B\subset\bS_R$.
\end{dfn}

\begin{rem}
By Lemma~\ref{lem:stochastic_X_not_uniform} it holds $X_{R,N}\in L^1_2(\bS_R)\subset L_n(\bS_R)$ almost surely. 
In particular, $Y^g_{R,N}$ is well-defined. Moreover, for positive measurable $F$ we have
\begin{equation}
 \int F(\phi)\, \mu^g_{R,N}(\rd\phi)
 =
 \frac{\bE F(X_{R,N})\exp(-Y^g_{R,N})}{\bE\exp(-Y^g_{R,N})}.
\end{equation} 
\end{rem}
\begin{lem}\label{lem:E_F_X}
Let $X$ be a real-valued random variable such that $X\geq 0$. Suppose that the function $F\,:\,[0,\infty)\to[0,\infty)$ is continuously differentiable and such that $F(0)=0$ and $F'\geq 0$. Then it holds 
\begin{equation}
 \bE F(X)=\int_0^\infty \mathbb{P}(X>t)\, F'(t)\,\rd t.
\end{equation} 
\end{lem}

\begin{lem}\label{lem:polynomial_bound}
There exists $A\in(0,\infty)$ depending only on the coefficients of the polynomial $\tau \mapsto P(\tau)$ such that for all $\tau\in\bR$ and $c\in(1,\infty)$ it holds $P(\tau,c)\geq \tau^n/2n-A \,c^{n/2}$.
\end{lem}
\begin{proof}
By the Young inequality for all $m \in \{0,1,\ldots, n-1\}$, $k\in\{0,1,\ldots, \lfloor\frac{m}{2}\rfloor$\}, $a\in \bR$ and $\delta \in (0,1)$ there exists $C\in (0, \infty)$ such that for all $c\in (1, \infty)$ and $\tau \in \bR$ it holds
\begin{equation}
a\,\tau^{m-2k}c^{k} \geq -\delta\,\tau^n - C\,c^{\frac{n}{2}}.
\end{equation}
To conclude we apply the above bound to all terms of the polynomial $P(\tau,c)$ but the term $\tau^n/n$ and choose $\delta\in (0,1)$ sufficiently small.
\end{proof}

\begin{prop}\label{prop:uv_limit}
Let $R\in\bN_+$ and $g\in C^\infty(\bS_R)$ satisfy the bounds~\eqref{eq:g_conditions}. There exist random variables $X_R\in \sD'(\bS_R)$, see Def. \ref{dfn:X_Y}, and $Y_R^g:=Y_R-X_R(g)^n/n\in\bR$ such that $\bE\exp(-Y_R^g)<\infty$ and for all bounded and continuous $F\,:\,\sD'(\bS_R)\to \bR$ and $p\in(0,\infty)$ it holds
\begin{equation}
 \lim_{N\to\infty} \bE F(X_{R,N}) \exp(-p\,Y_{R,N}^g) = \bE F(X_R) \exp(-p\,Y_R^g).
\end{equation}
\end{prop}
\begin{proof}
By Vitali's theorem it suffices to establish that $(F(X_{R,N}) \exp(-p\,Y_{R,N}^g))_{N\in\bN_+}$ converges in probability to $F(X_R) \exp(-p\,Y_R^g)$ and is uniformly integrable. The convergence in probability follows from Lemmas~\ref{lem:stochastic_X_uniform} and~\ref{lem:stochastic_Y}. To show uniform integrability it is enough to demonstrate that $(\bE\exp(-p\,Y_{R,N}^g))_{N\in\bN_+}$ is bounded for all $p\in(0,\infty)$.
By Lemma~\ref{lem:E_F_X} we have 
\begin{equation}\label{eq:density_bound}
\begin{aligned}
 \bE\exp(-p\,Y_{R,N}^g) &\leq 1+\bE( \exp(-p\,(Y^g_{R,N}\wedge0))-1) 
 \\
 &=1+\int_0^\infty \mathbb{P}(-p\,(Y^g_{R,N}\wedge 0)>t)\,\exp(t) \,\rd t
 \\
 &= 1+\int_0^\infty \mathbb{P}(-p\,Y^g_{R,N}>t)\, \exp(t) \,\rd t.
\end{aligned}
\end{equation}
By Lemma~\ref{lem:polynomial_bound} for every $R\in\bN_+$ there exists $A\in(0,\infty)$ such that for all $M\in\bN_+$ it holds $Y_{R,M}^g\geq -A\,c_{R,M}^{n/2}$. Consequently, by adding the latter inequality to $-Y^g_{R,N}> 2A\,c_{R,M}^{n/2}$, for every $R\in\bN_+$ there exist $c,C\in(0,\infty)$ such that for all $N,M\in\bN_+$ it holds
\begin{multline}
 \mathbb{P}(-Y^g_{R,N}> 2A\,c_{R,M}^{n/2}) 
 \leq 
 \mathbb{P}(|Y^g_{R,N}-Y^g_{R,M}|> A\, c_{R,M}^{n/2} ) \\
 \leq
 \exp(-c\,c_{R,M}\,A^{2/n}\,M^{1/n^2}) \,\bE \exp(c\,M^{1/n^2}\,|Y^g_{R,N}-Y^g_{R,M}|^{2/n})
 \leq
 C\, \exp(-M^{1/n^2}),
\end{multline}
where the last bound follows from Lemmas~\ref{lem:stochastic_X_uniform} and~\ref{lem:stochastic_Y}, the Nelson hypercontractivity estimate stated in Lemma~\ref{lem:nelson} and the estimate~\eqref{eq:bound_counterterm} for the counterterm $c_{R,M}$. As a result, by the bound~\eqref{eq:bound_counterterm} for every $R\in\bN_+$ and $p\in(0,\infty)$ there exist $c,C\in(0,\infty)$ such that for all $N\in\bN_+$ and $t\in[0,\infty)$ it holds
\begin{equation}
 \mathbb{P}(-p\,Y^g_{R,N}> t) \leq C\,\exp(-\exp(c\,t^{2/n})). 
\end{equation}
The above bound together with Eq.~\eqref{eq:density_bound} imply that $(F(X_{R,N}) \exp(-p\,Y^g_{R,N}))_{N\in\bN_+}$ is uniformly integrable. This finishes the proof.
\end{proof}

\section{Stochastic quantization}\label{sec:stochastic_quantization}

In order to show the existence of the infinite volume limit of the $P(\Phi)_2$ model and prove the bound~\eqref{eq:main_thm_bound} we have to establish appropriate bounds for moments of the regularized measure $\mu^g_{R,N}$ uniform in $R,N\in\bN_+$. To this end, we shall use the so-called parabolic stochastic quantization technique. We study a certain stochastic process evolving in fictitious time whose stationary distribution coincides with the Euclidean QFT measure. The process satisfies a non-linear stochastic PDE that is called the stochastic quantization equation. More specifically, to prove desired uniform bounds we apply the energy method, which relies on testing the equation by the solution itself and estimating the terms that are not positive. Because of the UV problem the stochastic quantization equation of the measure $\mu^g_{R,N}$, that is Eq.~\eqref{eq:weak_interacting} below, becomes singular in the limit $N\to\infty$. For this reason we cannot apply the energy method directly to Eq.~\eqref{eq:weak_interacting}. We use the so-called Da Prato-Debussche trick~\cite{daprato2003} that is based on the observation that the most singular part of the solution $\varPhi^g_{R,N}$ of Eq.~\eqref{eq:weak_interacting} coincides with the solution $Z_{R,N}$ of the stochastic quantization equation of the Gaussian measure $\nu_{R,N}$, that is Eq.~\eqref{eq:weak_free}. It turns out that Eq.~\eqref{eq:varPsi_pde}, which is satisfied by the process $\varPsi^g_{R,N}:=\varPhi^g_{R,N} - Z_{R,N}$, is not singular in the limit $N\to\infty$. The application of the energy method to~Eq.~\eqref{eq:varPsi_pde} is the subject of Sec.~\ref{sec:a_priori_bound}.

\begin{dfn}
For $R\in[1,\infty)$ we define $(W_R(t,\Cdot))_{t\in[0,\infty)}$ to be the cylindrical Wiener process on $L_2(\bS_R)$, see 
\cite[p. 53]{DaPratoZabczyk}.
\end{dfn}

\begin{dfn}
For $R,N\in[1,\infty)$ we set $Q_{R,N}:=(1-\Delta_R)(1-\Delta_R/N^2)^2$.
\end{dfn}

We study the following stochastic ODEs, which coincide with the stochastic quantization equations of the measures $\nu_{R,N}$ and $\mu_{R,N}^g$, respectively:
\begin{equation}\label{eq:weak_free}
 \rd Z_{R,N}(t,\Cdot) 
 =
 \sqrt{2}\, \rd W_R(t,\Cdot)- Q_{R,N} Z_{R,N}(t,\Cdot)\,\rd t,
\end{equation}
\begin{multline}\label{eq:weak_interacting}
 \rd \varPhi_{R,N}^g(t,\Cdot)
 =
 \sqrt{2}\, \rd W_R(t,\Cdot)- Q_{R,N} \varPhi_{R,N}^g(t,\Cdot)\,\rd t 
 \\
 -P'(\varPhi_{R,N}^g(t,\Cdot),c_{R,N})\,\rd t + (\varPhi_{R,N}^g(t,g))^{n-1}g\,\rd t,
\end{multline}
where $P'(\tau,c):=\partial_\tau P(\tau,c)$. The unique  mild  solution $Z_{R,N}\in C([0,\infty),L^1_2(\bS_R))$ 
of Eq.~\eqref{eq:weak_free} with the initial condition
\begin{equation}
 Z_{R,N}(0,\Cdot)=z_{R,N}\in L^1_2(\bS_R)
\end{equation}
is given by
\begin{equation}\label{eq:mild_free}
 Z_{R,N}(t,\Cdot) = \re^{-t Q_{R,N}} z_{R,N}+\int_0^t\re^{-(t-s)Q_{R,N}}\, \sqrt{2}\,\rd W_R(s,\Cdot)
\end{equation}
for $t\in[0,\infty)$, see e.g. \cite[Sec.~5.2]{DaPratoZabczyk}. 
By definition the mild solution of Eq.~\eqref{eq:weak_interacting} with the initial condition
\begin{equation}
 \varPhi_{R,N}^{g}(0,\Cdot)=\phi_{R,N}^g\in L^1_2(\bS_R)
\end{equation}
is the stochastic process $\varPhi_{R,N}^{g}\in C([0,\infty),L^1_2(\bS_R))$ such that for all $t\in[0,\infty)$ it holds
\begin{multline}\label{eq:mild_interacting}
 \varPhi_{R,N}^g(t,\Cdot) = \re^{-t Q_{R,N}} \phi_{R,N}^g
 \\
 +\int_0^t\re^{-(t-s)Q_{R,N}}\, \big(\sqrt{2}\,\rd W_R(s,\Cdot)
 -P'(\varPhi_{R,N}^g(s,\Cdot),c_{R,N})\rd s + (\varPhi_{R,N}^g(s,g))^{n-1}g\,\rd s\big). 
\end{multline}
The mild solution $\varPhi^g_{R,N}$ exists and is unique, {cf.~\cite[Sect.~5.5]{DaPratoZabczyk}}.

\begin{dfn}
By definition the stochastic processes $Z_{R,N},\varPhi^g_{R,N}\in C([0,\infty), L^1_2(\bS_R))$ are the unique solutions of Eq.~\eqref{eq:mild_free} and Eq.~\eqref{eq:mild_interacting}, respectively, with random initial data $z_{R,N}$ and $\phi^g_{R,N}$, respectively, such that $z_{R,N}$ and $\phi^g_{R,N}$ are independent of $(W_R(t,\Cdot))_{t\in[0,\infty)}$  and $\mathrm{Law}(z_{R,N},\phi_{R,N}^g)=\nu_{R,N}\times\mu^g_{R,N}$. We also define the process
\begin{equation}
 \varPsi^g_{R,N}:=\varPhi^g_{R,N} - Z_{R,N}\in C([0,\infty), L^1_2(\bS_R)).
\end{equation}
\end{dfn}

\begin{rem}
The processes $Z_{R,N},\varPhi^g_{R,N},\varPsi^g_{R,N}\in C([0,\infty), L^1_2(\bS_R))$ are well-defined  because the measures $\nu_{R,N}$ and $\mu^g_{R,N}$ are concentrated on $L^1_2(\bS_R)$.
\end{rem}

The following lemma expresses the fact that the measures $\nu_{R,N}$ and $\mu^{g}_{R,N}$ are invariant for Eq.~\eqref{eq:mild_free} and Eq.~\eqref{eq:mild_interacting}, respectively. 
\begin{lem} \label{lem:Law-Z-phi}
For all $t\in[0,\infty)$ it holds
\begin{equation}
 \mathrm{Law}(Z_{R,N}(t,\Cdot))=\nu_{R,N},
 \qquad
 \mathrm{Law}(\varPhi^g_{R,N}(t,\Cdot))=\mu^g_{R,N}.
\end{equation}
\end{lem}
\begin{proof}
See \cite[Sec.~8.6]{DaPratoZabczyk}. 
\end{proof}

\begin{lem}
It holds 
\begin{equation}
 \varPsi^g_{R,N}\in 
 C([0,\infty), L^1_2(\bS_R))\cap 
 C((0,\infty),L^3_2(\bS_R))\cap 
 C^1((0,\infty),L^{-3}_2(\bS_R))
\end{equation}
and the following equality
\begin{equation}\label{eq:varPsi_pde}
  \partial_t\varPsi^g_{R,N}  = -Q_{R,N}\varPsi^g_{R,N}- P'(\varPsi^g_{R,N}+Z_{R,N},c_{R,N}) +((\varPsi^g_{R,N}+Z_{R,N})(t,g))^{n-1}g
\end{equation}
is satisfied in $C((0,\infty),L^{-3}_2(\bS_R))$. 
\end{lem}
\begin{proof}
We first note that
\begin{multline}
 \varPsi^g_{R,N}(t,\Cdot) = \re^{-t Q_{R,N}}(\phi_{R,N}^g-z_{R,N})
 \\
 -
 \int_0^t \re^{-(t-s)Q_{R,N}} \big(P'(\varPhi^g_{R,N}(s,\Cdot),c_{R,N})-(\varPhi^g_{R,N}(s,g))^{n-1}g(\Cdot)\big)\,\rd s.
\end{multline}
It holds $\phi_{R,N}^g-z_{R,N}\in L_2^1(\bS_R)$ and $P'(\varPhi^g_{R,N},c_{R,N})\in C([0,\infty),L_2(\bS_R))$  almost surely by Lemma~\ref{lem:embedding_sphere}. The statement follows from the regularizing properties of the semigroup $(\re^{-t Q_{R,N}})_{t\in[0,\infty)}$.
\end{proof}

\begin{dfn}
For $l\in\{0,\ldots,n-2\}$ and $m\in\{l,\ldots,n-1\}$ we define
\begin{equation}
 a_{m,l}:=-a_{m+1}\,(m+1)!/(m-l)!l!,
\end{equation}
where $(a_m)_{m\in\{1,\ldots,n\}}$ are the coefficients of the polynomial $P(\tau)$. 
\end{dfn}

\begin{dfn}
By definition $Z_{R,N}^{:0:}:=1$ and for $m\in\{1,\ldots,n-1\}$ 
\begin{equation}
 Z_{R,N}^{:m:}:=\sum_{k=0}^{\floor{m/2}} \frac{(-1)^k m!}{(m-2k)!k!2^k} (c_{R,N})^k Z_{R,N}^{m-2k}.
\end{equation} 
\end{dfn}

Note that it holds
\begin{equation}
 P'(\varPsi^g_{R,N}+Z_{R,N},c_{R,N}) 
 = 
 (\varPsi^g_{R,N})^{n-1}
 -
 \sum_{l=0}^{n-2}\sum_{m=l}^{n-1}  a_{m,l}\, Z_{R,N}^{:m-l:} 
 \,(\varPsi^g_{R,N})^l,
\end{equation}
where $P'(\tau,c):=\partial_\tau P(\tau,c)$.  Consequently, Eq.~\eqref{eq:varPsi_pde} can be rewritten in the form
\begin{multline}\label{eq:varPsi_pde2}
 (\partial_t + Q_{R,N})\varPsi^g_{R,N} 
 + (\varPsi^g_{R,N})^{n-1}
 \\
 =
 \sum_{l=0}^{n-2}\sum_{m=l}^{n-1} a_{m,l}\, Z_{R,N}^{:m-l:} 
 \,(\varPsi^g_{R,N})^l
 +
 ((\varPsi^g_{R,N}+Z_{R,N})(\Cdot,g))^{n-1}g.
\end{multline}

\section{Stereographic projection}\label{sec:stereographic}

\begin{dfn}
For $R\in[1,\infty)$ we define 
\begin{equation}
 \bS_R:=\{\rx=(\rx_1,\rx_2,\rx_3)\in\bR^3\,|\, \rx_1^2+\rx_2^2+\rx_3^2=R^2\}.
\end{equation}
For $R\in[1,\infty)$ and $\rx,\ry\in\bS_R$ we denote by $d_R(\rx,\ry)$ the length of the shortest curve in $\bS_R\subset\bR^3$ connecting $\rx$ and $\ry$. For $R\in[1,\infty)$ we denote by $\rho_R$ the rotationally invariant measure on $\bS_R$ normalized such that $\rho_R(\bS_R)=4\pi R^2$. For $R\in[1,\infty)$ we denote by $\Delta_R$ the Laplace-Beltrami operator on $\bS_R$. We denote by $\Delta$ the Laplace operator on $\bR^2$.
\end{dfn}

\begin{dfn}\label{dfn:stereographic}
For $R\in[1,\infty)$ the map $\jmath_R\,:\bR^2\to\bS_R\setminus\{(0,0,R)\}\subset \bR^3$ is defined by (cf. Fig.~\ref{fig})
\begin{equation}
 \jmath_R(x_1,x_2):=\rx\equiv (\rx_1,\rx_2,\rx_3) = \frac{R\,(4R x_1,4R x_2,x_1^2+x_2^2-4R^2)}{4R^2+x_1^2+x_2^2}.
\end{equation} 
We call $(x_1,x_2)=x\in\bR^2$ the stereographic coordinates of $\rx\in\bS_R\setminus\{(0,0,R)\}$. 
We denote by $\jmath_R^*$ the pullback by $\jmath_R$.
We also set $w_R(x) := 16R^4/(4R^2+x_1^2+x_2^2)^2$.
\end{dfn}

\begin{figure}\label{fig}
\begin{center}
  \includegraphics[width=0.7\linewidth]{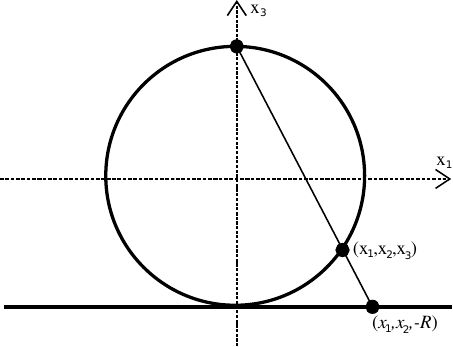}
  \caption{Stereographic projection} 
\end{center}
\end{figure}

\begin{rem}\label{rem:pulback_dist}
If $f\in C(\bS_R)$, then $\jmath_R^*f=f\circ\jmath_R\in C(\bR^2)$. Note that for $f\in C^\infty_\rc(\bR^2)$ the function $f\circ \jmath_R^{-1}\in C_\rc(\bS_R\setminus\{(0,0,R)\})$ has unique smooth extension to $\bS_R$. If $\phi\in\sD'(\bS_R)$, then $\jmath_R^*\phi\in\sS'(\bR^2)$ is defined by $\langle \jmath_R^*\phi,f\rangle := \langle \phi,(w_R^{-1}f)\circ \jmath_R^{-1}\rangle$ for all $f\in C^\infty_\rc(\bR^2)$. 
\end{rem}

\begin{rem}
In what follows, the function $w_R$ will play a prominent role. Note that the measure $\rho_R$ and the Laplace-Beltrami operator $\Delta_R$ on $\bS_R$ written in the stereographic coordinates take the following forms $w_R(x)\,\rd x$ and $w_R^{-1}(x)\Delta$. More precisely, the following identities
\begin{equation}
 \int_{\bS_R} f(\rx)\,\rho_R(\rd\rx) = \int_{\bR^2} \jmath_R^*f(x)\,w_R(x)\,\rd x,
 \qquad
 \jmath_R^*\Delta_R = w_R^{-1} \Delta\jmath_R^*\,
\end{equation}
are true.
\end{rem}

\section{A priori bound}\label{sec:a_priori_bound}

Note that Eq.~\eqref{eq:varPsi_pde2} takes the following form in the stereographic coordinates 
\begin{multline}\label{eq:spde_psi_stereographic}
 (\partial_t + (1- w_R^{-1}(x)\Delta)(1-w_R^{-1}(x)\Delta/N^2)^2)\jmath_R^*\varPsi^g_{R,N}(t,x) + (\jmath_R^*\varPsi^g_{R,N}(t,x))^{n-1}
 \\
 = 
 \sum_{l=0}^{n-2}\sum_{m=l}^{n-1} a_{m,l}\, \jmath_R^* Z_{R,N}^{:m-l:}(t,x)
 \,(\jmath_R^*\varPsi^g_{R,N}(t,x))^l
 \\
 +
 (\jmath_R^*(\varPsi^g_{R,N}+Z_{R,N})(t,w_R\jmath_R^* g))^{n-1}\jmath_R^*g(x)\,.
\end{multline}
In this section we prove an a priori bound by multiplying both sides of the above equation by $v_L\jmath_R^*\varPsi^g_{R,N}$, where $v_L$ is a suitable weight, and integrating over $\bR^2$. The bound is stated in the proposition below and is used in the next section to prove the existence of the infinite volume limit.

\begin{dfn}
Let $v_L:= \frac{1}{4\pi L^2} w_L^8$, where $L\in[1,\infty)$ is fixed as in Lemma~\ref{lem:weights_derivatives}.
\end{dfn}

\begin{rem}
The precise choice of the weight $v_L$ is not of much importance. It is convenient to use a weight that decays polynomially and express it in terms of the function $w_R$ introduced in Def.~\ref{dfn:stereographic}. The prefactor $\frac{1}{4\pi L^2}$ guarantees that the $L_1(\bR^2)$ norm of the weight is bounded by $1$ and the decay rate is chosen so that the estimate stated in Remark~\ref{rem:L_two_interms_of_L_n} is true.
\end{rem}

\begin{prop}\label{prop:energy_method}
There exist $\kappa\in(0,\infty)$, $C\in(0,\infty)$, $p\in[1,\infty)$ and a ball $B\subset\sS(\bR^2)$ with respect to some Schwartz semi-norm centered at the origin such that for all $t\in(0,\infty)$ and $R,N\in\mathbb{N}_+$, $R\geq L$, as well as all $g\in C^\infty(\bS_R)$, $w_R\jmath_R^* g\in B$, it holds
\begin{equation}
8\,\partial_t\|\jmath_R^* \varPsi^g_{R,N}(t,\Cdot)\|^2_{L_2(\bR^2,v_L^{1/2})}
 +
 \|\jmath_R^* \varPsi^g_{R,N}(t,\Cdot)\|^n_{L_n(\bR^2,v_L^{1/n})}
 \\
\leq C \sum_{k=0}^{n-1}\,
\|\jmath_R^*Z_{R,N}^{:k:}(t,\Cdot)\|^p_{L_p^{-\kappa}(\bR^2, v_L^{1/p})}\,.
\end{equation}
\end{prop}
\begin{proof}
After multiplying both sides of Eq.~\eqref{eq:spde_psi_stereographic} by $v_L\jmath_R^*\varPsi^g_{R,N}$, integrating over space and applying Lemma~\ref{lem:weights_derivatives} and Remark~\ref{rem:L_two_interms_of_L_n} we obtain
\begin{multline}
1/2~\partial_t\|\jmath_R^* \varPsi^g_{R,N}(t,\Cdot)\|^2_{L_2(\bR^2,v_L^{1/2})}
 +
  1/2~\|\vec\nabla\jmath_R^* \varPsi^g_{R,N}(t,\Cdot)\|^2_{L_2(\bR^2,v_L^{1/2})}   
\\
+
1/8~\|\jmath_R^* \varPsi^g_{R,N}(t,\Cdot)\|^n_{L_n(\bR^2,v_L^{1/n})}
\leq
R^{(1)}_{R,N}(t) +R^{(2)}_{R,N}(t),
\end{multline}
where
\begin{equation}
\begin{aligned}
R^{(1)}_{R,N}(t) &= \sum_{l=0}^{n-2}\sum_{m=l}^{n-1} \int_{\bR^2} a_{m,l}\,v_L(x)\,\jmath_R^*Z_{R,N}^{:m-l:}(t,x)
 \,(\jmath_R^*\varPsi^g_{R,N}(t,x))^{l+1} \,\rd x,
\\
R^{(2)}_{R, N}(t) &= 
(\jmath_R^*(\varPsi^g_{R,N}+Z_{R,N})(t,w_R\jmath_R^* g))^{n-1}\,
(\jmath_R^*\varPsi^g_{R,N})(t,v_L\jmath_R^*g).
\end{aligned}
\end{equation}
By Lemma~\ref{lem:Z_Psi_n} for every $\delta_1\in(0,1)$ there exists $C_1\in(0,\infty)$ such that
\begin{multline}
|R_{R,N}^{(1)}(t)|\leq C_1 \sum_{k=0}^{n-1}\,\|\jmath_R^*Z_{R,N}^{:k:}(t,\Cdot)\|^{p}_{L_p^{-\kappa}(\bR^2,v_L^{1/p})}
\\
+\delta_1\,\|\nabla\jmath_R^*\varPsi^g_{R,N}(t,\Cdot)\|^2_{L_2(\bR^2, v_L^{1/2})}
+ \delta_1\,\|\jmath_R^*\varPsi^g_{R,N}(t,\Cdot) \|^n_{L_n(\bR^2, v_L^{1/n})},\,
\end{multline}
where $k=0$ term of the sum above is a constant. Furthermore, by H{\"o}lder's inequality and elementary estimates there exists $C_2\in(0,\infty)$ such that for all $\delta_2\in(0,1)$ it holds
\begin{equation}
|R^{(2)}_{R,N}(t)| \leq C_2\,\delta_2^n\,
\|\jmath_R^*\varPsi^g_{R,N}(t,\Cdot)\|^n_{L_n(\bR^2, v_L^{1/n})}
+
C_2\,\delta_2^n\,\|\jmath_R^*Z_{R,N}\|^n_{L_n^{-\kappa}(\bR^2, v_L^{1/n})}
\end{equation}
provided
\begin{equation}
 \|v_L^{(n-1)/n}\jmath_R^*g\|_{L_{n/(n-1)}(\bR^2)}
 \vee
 \|v_L^{-1/n}w_R\jmath_R^*g\|_{L_{n/(n-1)}(\bR^2)}
 \vee
 \|v_L^{-1/n}w_R\jmath_R^*g\|_{L^\kappa_{n/(n-1)}(\bR^2)}
 \leq \delta_2.
\end{equation}
We choose $\delta_1,\delta_2$ such that $\delta_1\leq1/2$ and $\delta_1+C_2\,\delta_2^n\leq1/16$. This finishes the proof.
\end{proof}

\begin{lem}\label{lem:weights_derivatives}
There exists $L\in[1,\infty)$ such that for all $R\in[L,\infty)$ it holds 
\begin{enumerate}
\item[(A)] $\langle\Psi, v_L (-w^{-1}_R \Delta)\Psi\rangle_{L_2(\bR^2)} 
\geq
1/2\,\|\vec\nabla\Psi\|_{L_2(\bR^2,w^{-1/2}_Rv_L^{1/2})}^2
-1/8\,\|\Psi\|^2_{L_2(\bR^2,w_R^{-1/2} v_L^{1/2})}$,

\item[(B)] $\langle\Psi, v_L (-w^{-1}_R\Delta)^2 \Psi\rangle_{L_2(\bR^2)}
\geq 
1/2\,\|\Delta\Psi\|_{L_2(\bR^2,w^{-1}_Rv_L^{1/2})}^2
-1/8\,\|\Psi\|^2_{L_2(\bR^2,w_R^{-1} v_L^{1/2})}$,

\item[(C)] $\langle\Psi, v_L (-w^{-1}_R\Delta)^3 \Psi\rangle_{L_2(\bR^2)}
\geq 
1/2\,\|\vec\nabla\Delta\Psi\|_{L_2(\bR^2,w^{-3/2}_Rv_L^{1/2})}^2
-1/8\,\|\Psi\|^2_{L_2(\bR^2,w_R^{-3/2} v_L^{1/2})}$.
\end{enumerate}
\end{lem}
\begin{proof}
There exists $C\in(0,\infty)$ such that for all $L\in[1,\infty)$, $R\in[L,\infty)$ it holds
\begin{equation}
|\vec\nabla w^{-1/2}_R|\leq C\,w_R^{-1/2}/L,
\qquad
|\vec\nabla v^{1/2}_L|\leq C\,v_L^{1/2}/L.
\end{equation}
This gives readily (A) by integrating by parts in  $\|\vec\nabla\Psi\|_{L_2(\bR^2,w^{-1/2}_Rv_L^{1/2})}^2$  applying the
Leibniz rule and the Young inequality. Estimates (B) and (C) are obtained analogously, with the help of the following auxiliary inequalities
\begin{equation}
 \|\vec\nabla\Psi\|_{L_2(\bR^2,w^{-1}_Rv_L^{1/2})}^2
 \leq 
 2\|\Delta\Psi\|_{L_2(\bR^2,w^{-1}_Rv_L^{1/2})}^2
 +
 2\|\Psi\|^2_{L_2(\bR^2,w_R^{-1} v_L^{1/2})}\,,
\end{equation}
\begin{equation}
 \|\vec\nabla\Psi\|_{L_2(\bR^2,w^{-3/2}_Rv_L^{1/2})}^2
 +
 \|\Delta\Psi\|_{L_2(\bR^2,w^{-3/2}_Rv_L^{1/2})}^2
 \leq 
 4\|\vec\nabla\Delta\Psi\|_{L_2(\bR^2,w^{-3/2}_Rv_L^{1/2})}^2
 +
 4\|\Psi\|^2_{L_2(\bR^2,w_R^{-3/2} v_L^{1/2})}
\end{equation}
valid for sufficiently big $L\in[1,\infty)$ and all $R\in[L,\infty)$. The latter inequalities are proven by the same token as (A).
\end{proof}

\begin{rem}\label{rem:L_two_interms_of_L_n}
For all $L\in[1,\infty)$, $R\in[L,\infty)$ and $p\in\{1,2,3\}$ it holds
\begin{equation}
 \|\Psi\|_{L_2(\bR^2,w_R^{-p/2} v_L^{1/2})}
 \leq
 \|w_R^{-p/2} v_L^{(n-2)/2n}\|_{L_{2n/(n-2)}(\bR^2) }
 \, 
 \|\Psi\|_{L_n(\bR^2,v_L^{1/n})}
 \leq
 \|\Psi\|_{L_n(\bR^2,v_L^{1/n})}.
\end{equation} 
\begin{equation}
\begin{aligned}
 \|\Psi\|_{L_2(\bR^2,w_R^{-p/2} v_L^{1/2})}
 &\leq
 \|w_R^{-p/2} v_L^{(n-2)/2n}\|_{L_{2n/(n-2)}(\bR^2) }
 \,
 \|\Psi\|_{L_n(\bR^2,v_L^{1/n})}
 \\
 &\leq
 \|\Psi\|_{L_n(\bR^2,v_L^{1/n})}.
\end{aligned}
\end{equation}
\end{rem}

\section{Tightness}\label{sec:tightness}

\begin{prop}\label{prop:tightness}
Let $\kappa\in(0,\infty)$. There exists a ball $B\subset\sS(\bR^2)$ with respect to some Schwartz semi-norm centered at the origin and a constant $C\in(0,\infty)$ such that for all $R\in\bN_+$, $R\geq L$,  $N\in \bN_+$ and all $g\in C^\infty(\bS_R)$, $w_R\jmath_R^* g\in B$, it holds
\begin{equation}
 \int\|\jmath^*_R\phi\|^n_{L_n^{-\kappa}(\bR^2,v_L^{1/n})}\,\mu^g_{R,N}(\rd\phi)\leq C.
\end{equation}
\end{prop}
\begin{rem}\label{rem:tightness}
By Proposition~\ref{prop:uv_limit} with $F=1$ and Lemma~\ref{lem:stochastic_bound_infinite_volume} we obtain
\begin{equation}
 \int\|\jmath_R^*\phi\|^n_{L_n^{-\kappa}(\bR^2,v_L^{1/n})}\,\mu_R^g(\rd\phi)
 =
 \lim_{N\to\infty}\int\|\jmath_R^*\phi\|^n_{L_n^{-\kappa}(\bR^2,v_L^{1/n})}\,\mu_{R,N}^g(\rd\phi)
 \leq C\,.
\end{equation} 
By Theorem~\ref{thm:embedding}~(C) the embedding ${L^{-\kappa}_n(\bR^2, v_L^{1/n})}\to L^{-2\kappa}_n(\bR^2, v_L^{2/n})$ is compact. As a result, by Lemma~\ref{lem:tightness_criterion} the sequence of measures $(\jmath_R^*\sharp\mu_R)_{R\in\bN_+}$ on $L^{-2\kappa}_n(\bR^2, v_L^{2/n})$ is tight and by Prokhorov's theorem it has a weakly convergent subsequence.
\end{rem}

\begin{proof}
Recall from Lemma~\ref{lem:Law-Z-phi} that $\mathrm{Law}(\varPhi^g_{R,N}(t,\Cdot))=\mu_{R,N}$ for all $t\in[0,\infty)$. Hence,
\begin{equation}
\int\|\jmath^*_R\phi\|^n_{L_n^{-\kappa}(\bR^2,v_L^{1/n})}\,\mu^g_{R,N}(\rd\phi)
=
\bE\|\jmath_R^*\varPhi^g_{R,N}(t,\Cdot)\|^n_{L_n^{-\kappa}(\bR^2,v_L^{1/n})}\,.
\end{equation}
Since $\mathrm{Law}(X_{R, N})=\mathrm{Law}(Z_{R,N}(t,\Cdot))$ for all $t\in[0,\infty)$ by Lemma~\ref{lem:stochastic_bound_infinite_volume} and
Proposition~\ref{prop:energy_method} we have
\begin{equation}\label{eq: first priori bound}
8\,\partial_t\bE\|\jmath_R^* \varPsi^g_{R,N}(t,\Cdot)\|^2_{L_2(\bR^2,v_L^{1/2})}
 +
\bE \|\jmath_R^* \varPsi^g_{R,N}(t,\Cdot)\|^n_{L_n(\bR^2,v_L^{1/n})}
\leq C_1 \,
\end{equation}
for some constant $C_1\in(0,\infty)$ independent of $g$, $R,N$ and $t$. The above inequality implies that for all $T\in(0,\infty)$ it holds
\begin{multline}\label{eq:tightness_time_average}
\frac{1}{T}\int_0^T\bE \|\jmath_R^* \varPsi^g_{R,N}(t,\Cdot)\|^n_{L_n(\bR^2,v_L^{1/n})}\rd t
\\
\leq C_1-\frac{8}{T}\bE\|\jmath_R^*\varPsi^g_{R,N}(T,\Cdot)\|^2_{L_2(\bR^2, v_L^{1/2})}
+\frac{8}{T}\bE\|\jmath_R^*\varPsi^g_{R,N}(0,\Cdot)\|^2_{L_2(\bR^2, v_L^{1/2})}
\leq C_1 + \frac{C_{R,N}}{T}\,,
\end{multline}
where
\begin{equation}
 C_{R,N}:=8\,\bE\|\jmath_R^*\varPsi^g_{R,N}(0,\Cdot)\|^2_{L_2(\bR^2, v_L^{1/2})}\leq 8\,\bE\|\varPsi^g_{R,N}(0,\Cdot)\|^2_{L_2(\bS_R)} < \infty
\end{equation}
for every $R,N\in\bN_+$ and $R\geq L$. Using the fact that $\jmath_R^*\varPhi^g_{R,N}$ and $\jmath_R^*Z_{R,N}$ are stationary in time one deduces that 
\begin{multline}\label{eq: bound for expectation phi}
\bE\|\jmath_R^*\varPhi^g_{R,N}(0, \Cdot)\|^n_{L^{-\kappa}_n(\bR^2, v_L^{1/n})}
=
\frac{1}{T}\int_0^T \bE\|\jmath_R^*\varPhi^g_{R,N}(t, \Cdot)\|^n_{L^{-\kappa}_n(\bR^2, v_L^{1/n})}\rd t
\\
\leq 
{c}\,\bE\|\jmath_R^*Z_{R,N}(0,\Cdot)\|^n_{L^{-\kappa}_n(\bR^2, v_L^{1/n})}+ \frac{c}{T}\int_0^T\bE\|\jmath_R^*\varPsi^g_{R,N}(t, \Cdot)\|^n_{L^{-\kappa}_n(\bR^2, v_L^{1/n})}\rd t\,,
\end{multline}
where $c=2^{n-1}$. By Lemma~\ref{lem:stochastic_bound_infinite_volume} there exists $C_2\in(0,\infty)$ such that for all $R,N\in\bN_+$ it holds
\begin{equation}
 \bE\|\jmath_R^*Z_{R,N}(0,\Cdot)\|^n_{L^{-\kappa}_n(\bR^2, v_L^{1/n})}\leq C_2.
\end{equation}
Combining the bounds proved above we obtain
\begin{equation}
 \bE\|\jmath_R^*\varPhi^g_{R,N}(0, \Cdot)\|^n_{L^{-\kappa}_n(\bR^2, v_L^{1/n})}\leq c\,C_1+c\,C_2 + \frac{c\,C_{R,N}}{T}
\end{equation}
for all $T\in(0,\infty)$. Choosing $T=C_{R,N}$ concludes the proof.
\end{proof}

\section{Integrability}\label{sec:integrability}

\begin{prop}\label{prop:integrability}
There exists a ball $B\subset\sS(\bR^2)$ with respect to some Schwartz semi-norm centered at the origin such that for all $f\in B$ the bound~\eqref{eq:main_thm_bound} holds true.
\end{prop}
\begin{proof}
It follows from properties of the stereographic coordinates that for all $f\in\sS(\bR^2)$ and $R\in\bN_+$ there exists $g_R\in C^\infty(\bS_R)$ such that $w_R\jmath_R^* g_R=f$. Let $B$ be contained in the ball from the statement of Proposition~\ref{prop:tightness} and suppose that $f\in B$. Note that for arbitrary $\phi\in\sD'(\bS_R)$ it holds
\begin{equation}\label{eq:integrability_g_identity}
 \phi(g_R)=(\jmath^*_R\phi)(w_R\jmath^*_Rg_R) = (\jmath^*_R\phi)(f).
\end{equation}  
Then  by Lemma~\ref{lem:mu_f_bound} it holds 
\begin{equation}
 \int \exp\left(\phi(g_R)^n/n\right)\,\mu_{R,N}(\rd\phi) 
 \\
 \leq 
 \exp\left(\frac{1}{n}\int \phi(g_R)^n\,\mu^{g_R}_{R,N}(\rd\phi)\right)\,.
\end{equation}
Note that the expression on the LHS is integrable by Lemma~\ref{lem:mu_measure_well_defined}. The identity~\eqref{eq:integrability_g_identity}, H{\"o}lder's inequality and Proposition~\ref{prop:tightness} yield
\begin{multline}
\int \phi(g_R)^n\, \mu_{R,N}^{g_R}(\rd\phi) 
\leq  \hat C\,
\|v_L^{-1/n}w_R\jmath^*_Rg_R\|^n_{L^{\kappa}_{n/(n-1)}(\bR^2)}
\,
\int\|\jmath_R^*\phi\|^n_{L_n^{-\kappa}(\bR^2, v_L^{1/n})}\,
\mu_R^{g_R}(\rd\phi)
\\
\leq\,C\,\, \|v_L^{-1/n}w_R\jmath^*_Rg_R\|^n_{L^{\kappa}_{n/(n-1)}(\bR^2)}
\end{multline}
for some constants $\hat C,C\in(0,\infty)$ independent of $R,N$ and $g_R$. Choosing the ball $B$ so that $\|v_L^{-1/n}f\|^n_{L^{\kappa}_{n/(n-1)}(\bR^2)}\leq n/2C$ for all $f\in B$ by the above inequalities and Proposition~\ref{prop:uv_limit} we obtain
\begin{equation}
 \int \exp\left((\jmath^*_R\phi)(f)^n/n\right)\,\mu_{R}(\rd\phi)
 =
 \lim_{N\to\infty}\int \exp\left(\phi(g_R)^n/n\right)\,\mu_{R,N}(\rd\phi)
 \leq 2\,.
\end{equation}
This concludes the proof.
\end{proof}

\begin{lem}\cite[Lemma~A.7]{BarashkovSinh}\label{lem:mu_f_bound}
Let $(\Omega,\cF,\mu)$ be a probability space, $F:\, \Omega \to \bR$ be a measurable function such that $\exp(F) \in L_1(\Omega,\mu)$ and
\begin{equation}
\mu^F(\rd\phi) := \frac{\exp(F(\phi))\,\mu(\rd\phi)}{\int \exp(F(\phi)) \,\mu(\rd\phi)}\,.
\end{equation}
It holds
\begin{equation}
 \int \exp(F(\phi))\,\mu(d\phi) \leq \exp\left(\int F(\phi)\,\mu^F(d\phi)\right)\,.
\end{equation}
\end{lem}

\section{Reflection positivity}\label{sec:reflection}

In this section, in Proposition~\ref{prop:reflection_positivity},  we establish the reflection positivity of every accumulation point of the sequence $(\jmath_R^*\sharp\mu_R)_{R\in\bN_+}$ of measures on $\sS'(\bR^2)$. To this end, we leverage the fact that the finite volume measure $\mu_R$ on $\sD'(\bS_R)$ is reflection positive.

\begin{dfn}
For all $R,N\in[1,\infty)$ we set
\begin{equation}
\bS_{R,N}^\pm:=\{(\rx_1,\rx_2,\rx_3)\in\bS_R\,|\,\pm\rx_1>1/N\},
\quad
\bS_{R,N} := \bS_{R,N}^+\cup \bS_{R,N}^-,
\quad
\bS_R^\pm:=\cup_{N\in[1,\infty)} \bS_{R,N}^\pm.
\end{equation}
\end{dfn}

\begin{dfn}\label{dfn:cylindircal_functionals}
Let $R\in[1,\infty)$. A functional $F\,:\sD'(\bS_R) \to \bC$ is called cylindrical iff there exists $k\in\bN_+$, $G\in C^\infty_{\mathrm{b}}(\bR^k)$ and $f_l\in  C^\infty_\rc(\bS_R):= C^\infty(\bS_R)$, $l\in\{1,\ldots,k\}$, such that
\begin{equation}\label{eq:cylindircal}
 F(\phi) = G(\phi(f_1),\ldots,\phi(f_k)).
\end{equation}
The algebra of cylindrical functions is denoted by $\cF_R$. The subalgebras of $\cF_R$ consisting of functionals of the form~\eqref{eq:cylindircal} with $\supp\,f_l\subset \bS_R^\pm$, $l\in\{1,\ldots,k\}$, or such that $\supp\,f_l\subset \bS_{R,N}^\pm$, $l\in\{1,\ldots,k\}$, are denoted by $\cF^\pm_R$ and $\cF^\pm_{R,N}$, respectively. The definitions of $\cF$ and $\cF^\pm$ are analogous to the definitions of $\cF_R$ and $\cF_R^\pm$ with $\bS_R$ and $\bS_R^\pm$ replaced by $\bR^2$ and the half-plane $\{(x_1,x_2)\in\bR^2|\pm x_1>0\}$, respectively. 
\end{dfn}

\begin{dfn}
Let $R\in[1,\infty)$. For $f\in C^\infty(\bS_R)$ we define $\Theta_R f\in C^\infty(\bS_R)$ by the formula $(\Theta_R f)(\rx_1,\rx_2,\rx_3):=f(-\rx_1,\rx_2,\rx_3)$. For $\phi\in\sD'(\bS_R)$ we define $\Theta_R \phi\in \sD'(\bS_R)$ by the formula $\langle\Theta_R\phi,f\rangle:=\langle\phi,\Theta_Rf\rangle$ for all $f\in C^\infty(\bS_R)$. For $f\in C^\infty(\bR^2)$ we define $\Theta f\in C^\infty(\bR^2)$ by the formula $(\Theta f)(x_1,x_2):=f(-x_1,x_2)$. For $\phi\in\sS'(\bR^2)$ we define $\Theta \phi\in \sS'(\bR^2)$ by the formula $\langle\Theta\phi,f\rangle:=\langle\phi,\Theta f\rangle$ for all $f\in\sS(\bR^2)$.
\end{dfn}

\begin{rem}
Note that $\jmath_R^* \Theta_R \phi = \Theta \jmath_R^* \phi$ for all $\phi\in\sD'(\bS_R)$. 
\end{rem}

\begin{prop}\label{prop:reflection_positivity}
Let $\mu$ be a weak limit of a subsequence of the sequence of measures $(\jmath_R^*\sharp\mu_R)_{R\in\bN_+}$ on $\sS'(\bR^2)$. For all $F\in\cF^+$ it holds
$
\int \overline{F(\Theta\phi)} F(\phi)\,\mu(\rd\phi) \geq 0
$.
\end{prop}
\begin{proof}
It is enough to prove that 
\begin{equation}
 \int \overline{F(\Theta\phi)} F(\phi)\,(\jmath_R^*\sharp\mu_R)(\rd\phi) 
 = 
 \int \overline{F(\Theta \jmath_R^*\phi)} F(\jmath_R^*\phi)\,\mu_R(\rd\phi)
 = 
 \int \overline{F(\jmath_R^*\Theta_R\phi)} F(\jmath_R^*\phi)\,\mu_R(\rd\phi)
 \geq 0
\end{equation} 
for all $R\in\bN_+$ and $F\in\cF^+$. By Def.~\ref{dfn:cylindircal_functionals} and Remark~\ref{rem:pulback_dist} for every $F\in\cF^+$ it holds $F\circ\jmath_R^*\in\cF^+_R$. Hence, the last bound above follows from the reflection positivity of the measure $\mu_R$.
\end{proof}

For completeness, we prove below the reflection positivity of the measure~$\mu_R$ on $\sD'(\bS_R)$, which is stated in Lemma~\ref{lem:reflection_positivity}~(D). Note that the UV cutoff in the definition of the measures $\mu_{R,N}$, introduced in Sec.~\ref{sec:uv_limit}, breaks the reflection positivity, cf. \cite{ABRSS18}. For this reason, in this section we work with a different UV cutoff. We introduce a free field $\hat X_{R,N}$ with a UV cutoff that preserves the reflection positivity, see Lemma~\ref{lem:reflection_positivity}~(B), and show that the measure $\mu_R$ can be approximated, see Lemma~\ref{lem:UV_convergence_tilde} and Eq.~\eqref{eq:positivity_mu_R_approximation}, by measures with a UV cutoff that are reflection positive, see Lemma~\ref{lem:reflection_positivity}~(C). 

\begin{dfn}\label{def:hat-K}
Fix $h\in C^\infty(\bR)$ such that $\supp\,h\subset(-1,1)$, $h=1$ on $[-1/2,1/2]$ and $2\pi\int_0^\infty h(\theta)\,\theta\,\rd\theta =1$. For $R,N\in\bN_+$ the operator $\hat K_{R,N}:\,L_2(\mathbb{S}_R)\to L_2(\mathbb{S}_R)$ is defined by its integral kernel
\begin{equation}
 \hat K_{R,N}(\rx,\ry):=N^2 h(N\,d_R(\rx,\ry)).
\end{equation}
\end{dfn}

\begin{rem}
Note that formally for $R=\infty$ we have $\bS_R=\bR^2$ and $d_R(\rx,\ry)=|\rx-\ry|$ as well as $\int_{\bR^2}\hat K_{\infty,N}(x,y)\,\rd y=2\pi\int_0^\infty h(\theta)\,\theta\,\rd\theta =1$.
\end{rem}

\begin{dfn}\label{dfn:hat_X_hat_Y}
Let $\hat c_{R,N}:=\mathrm{Tr}(\hat K_{R,N}G_R\hat K_{R,N})/4\pi R^2$. By definition $\hat X_{R,N}:=\hat K_{R,N}X_R$,
\begin{equation}
\begin{gathered}
 \hat X_{R,N}^{:m:}:=\sum_{k=0}^{\floor{m/2}} \frac{(-1)^k m!}{(m-2k)!k!2^k} (\hat c_{R,N})^k \hat X_{R,N}^{m-2k},
 \qquad
 \hat X_{R,N}^{:m:}(h):=\int_{\bS_R} \hat X_{R,N}^{:m:}(\rx)\,h(\rx)\,\rho_R(\rd \rx),
 \\
 \hat Y_{R,N}:=\sum_{m=0}^n a_m \hat X_{R,N}^{:m:}(1_{\bS_R}),
 \qquad
 \tilde Y_{R,N}^\pm:=\sum_{m=0}^n a_m \hat X_{R,N}^{:m:}(1_{\bS^\pm_{R,N}}),
 \qquad
 \tilde Y_{R,N} := \tilde Y^+_{R,N} + \tilde Y^-_{R,N},
\end{gathered} 
\end{equation}
where $h\in L_\infty(\bS_R)$ and $1_B$ denotes the characteristic function of the set $B$.
\end{dfn}

\begin{rem}
Note that $\hat X_{R,N}$ introduced above and $X_{R,N}$ introduced in Def.~\ref{dfn:X_Y} are free fields on $\bS_R$ with different UV cutoffs. We use the same symbol $N\in\bN_+$ to denote both cutoffs.
\end{rem}

\begin{rem}
By Lemma~\ref{lem:stochastic_X_not_uniform} it holds $\hat X_{R,N}\in L^1_2(\bS_R)\subset L_n(\bS_R)$ almost surely. In particular, $\hat Y_{R,N}$, $\tilde Y_{R,N}$ are well-defined. Moreover, there exists $C\in(0,\infty)$ such that for all $N,R\in\bN_+$ it holds $|\hat c_{R,N}-1/2\pi\,\log N|\leq C$ by the bound~\eqref{eq:bound_counterterm} and Remark~\ref{rem:hat_c}.
\end{rem}

\begin{lem}\label{lem:UV_convergence_tilde}
For all $R\in\bN_+$ and all bounded and continuous $F\,:\,\sD'(\bS_R)\to \bR$ it holds
\begin{equation}
 \lim_{N\to\infty} \bE F(\hat X_{R,N}) \exp(-\tilde Y_{R,N}) = \lim_{N\to\infty} \bE F(X_{R,N}) \exp(-Y_{R,N}). 
\end{equation}
\end{lem}
\begin{proof}
The proof follows the strategy of the proof of Proposition~\ref{prop:uv_limit}. By Lemmas~\ref{lem:stochastic_X_uniform} and~\ref{lem:stochastic_Y} the sequences $(F(\hat X_{R,N}) \exp(-\tilde Y_{R,N}))_{N\in\bN_+}$ and $(F(X_{R,N}) \exp(-Y_{R,N}))_{N\in\bN_+}$ converge in probability to $F(X_R) \exp(-Y_R)$. To conclude we show that the above-mentioned sequences are uniformly integrable by repeating verbatim the argument from the proof of Proposition~\ref{prop:uv_limit}.
\end{proof}

\begin{lem}\label{lem:reflection_positivity}
The following statements hold true for all $R,N\in\bN_+$:
\begin{enumerate}
 \item[(A)] If $F\in\cF_R^+$, then $\bE \overline{F(\Theta_R X_R)} F(X_R)\geq 0$.
 \item[(B)] If $F\in\cF_{R,N}^+$, then $\bE \overline{F(\Theta_R \hat X_{R,N})} F(\hat X_{R,N})\geq 0$.
 \item[(C)] If $F\in\cF_{R,N}^+$, then $\bE \overline{F(\Theta_R \hat X_{R,N})} F(\hat X_{R,N})\exp(-\tilde Y_{R,N})\geq 0$.
 \item[(D)] For all $F\in\cF^+_R$ it holds
$
\int \overline{F(\Theta_{R}\phi)} F(\phi)\,\mu_R(\rd\phi) \geq 0
$.
\end{enumerate}
\end{lem}
\begin{proof}
For the proof of Item~(A) see~\cite[Theorem~2]{dimock2004}. To prove Item~(B), observe that
\begin{equation}
\bE \overline{F(\Theta_R \hat X_{R,N})} F(\hat X_{R,N}) = \bE \overline{F( \hat K_{R,N} (\Theta_R X_R))} F(\hat K_{R,N} X_R)\,,
\end{equation}
where we have used the fact that $\Theta_R \hat K_{R,N} X_R = \hat K_{R,N} \Theta_R X_R$. By the support property of the kernel $\hat K_{R,N}(\rx,\ry)$ if $F\in\cF^+_{R,N}$, then the functional $\phi\mapsto F(\hat K_{R,N}\phi)$ belongs to $\mathcal{F}_R^+$. Consequently, the statement follows from Item~(A). To prove Item~(C), note that
\begin{equation}
\bE \overline{F(\Theta_R \hat X_{R,N})} F(\hat X_{R,N})\exp(-\tilde Y_{R,N})
=
\bE \overline{F(\Theta_R \hat X_{R,N})\exp(-\tilde Y^-_{R,N})}F(\hat X_{R,N})\exp(-\tilde Y^+_{R,N})\,.
\end{equation}
Denote $H(\hat X_{R,N}) := F(\hat X_{R,N}) \exp(-\tilde Y_{R,N}^+)$. It holds
\begin{equation}
\bE \overline{F(\Theta_R \hat X_{R,N})} F(\hat X_{R,N})\exp(-\tilde Y_{R,N})=\bE\overline{ H(\Theta_R\hat X_{R,N}) }  H(\hat X_{R,N})\,.
\end{equation}
The RHS of the above equality can be approximated by a similar expression with $H$ replaced by some functional belonging to $\cF_{R,N}^+$. As a result, Item~(C) follows from Item~(B). Let us turn to the proof of Item~(D). First note that for any $F\in\cF^+_R$ there exists $M\in\bN_+$ such that $F\in\cF^+_{R,M}$. Hence, it suffices to show that $\int \overline{F(\Theta_R\phi)} F(\phi)\,\mu_R(\rd\phi) \geq 0$ for all $R,M\in\bN_+$ and $F\in\cF^+_{R,M}$. To establish this claim we note that by Lemma~\ref{lem:UV_convergence_tilde}
\begin{equation}\label{eq:positivity_mu_R_approximation}
\int \overline{F(\Theta_R\phi)} F(\phi)\,\mu_R(\rd\phi)
=
\lim_{N\to\infty} \frac{\bE\overline{F(\Theta_R \hat X_{R,N})} F(\hat X_{R,N})\exp(-\tilde Y_{R,N})}{\bE \exp(-\tilde Y_{R,N}) }
\end{equation}
and use Item~(C) together with the fact that $\cF^+_{R,M}\subset\cF^+_{R,N}$ for all $N\geq M$.
\end{proof}

\section{Euclidean invariance}\label{sec:Euclidean}

In this section we establish the invariance under the Euclidean transformations of the plane of every accumulation point $\mu$ of the sequence $(\jmath_R^*\sharp\mu_R)_{R\in\bN_+}$ of measures on $\sS'(\bR^2)$. We use the fact that for every $R\in\bN_+$ the measure $\mu_R$ is invariant under the rotations of the sphere $\bS_R$. The proof of the rotational invariance of $\mu$ is straight-forward as the rotations $\cR_{R,\alpha}$ of the sphere $\bS_R$ around the $\rx_3$ axis are mapped under the stereographic projection to the rotations $\cR_{\alpha}$ of the plane $\bR^2$ around the origin. Hence, for every $R\in\bR_+$ the measure $\jmath_R^*\sharp\mu_R$ on $\sS'(\bR^2)$ is invariant under the rotations around the origin and the same is true for every accumulation point $\mu$. The proof of the translational invariance of $\mu$ is more complicated. There is no rotation of the sphere $\bS_R$ that is mapped under the stereographic projection to the translation $\cT_{\alpha}$ of the plane $\bR^2$ in the $x_1$ direction. In particular, for every $R\in\bR_+$ the measure $\jmath_R^*\sharp\mu_R$ on $\sS'(\bR^2)$ is not invariant under the translations. In order to establish the translational invariance of an accumulation point $\mu$ we first prove that the rotations $\cT_{R,\alpha}$ of the sphere $\bS_R$ around the $\rx_2$ axis are mapped under the stereographic projection to certain transformations $\cS_{R,\alpha}$ of the plane $\bR^2$ and subsequently show that the transformations $\cS_{R,\alpha}$ converge to the translations $\cT_{\alpha}$ of the plane $\bR^2$ in the $x_1$ direction as $R\to\infty$.

\begin{dfn}
For $\alpha\in\bR$ the maps $\cR_{\alpha},\cT_{\alpha}\,:\,\bR^2\to\bR^2$ are defined by
\begin{equation}
 \cR_\alpha(x_1,x_2) 
 := (x_1\cos\alpha+x_2\sin\alpha,x_1\sin\alpha-x_2\cos\alpha),
 \qquad
 \cT_\alpha(x_1,x_2) := (x_1+\alpha,x_2),
\end{equation}
For $R\in\bN_+$, $\alpha\in\bR$ the maps $\cR_{R,\alpha},\cT_{R,\alpha}\,:\,\bS_R\to\bS_R$ are defined by
\begin{equation}
\begin{gathered}
 \cR_{R,\alpha}(\rx) 
 = (\rx_1\cos\alpha+\rx_2\sin\alpha,\rx_1\sin\alpha-\rx_2\cos\alpha,\rx_3),
 \\
 \cT_{R,\alpha}(\rx) 
 = (\rx_1\cos(\alpha/R)-\rx_3\sin(\alpha/R),\rx_2,\rx_1\sin(\alpha/R)+\rx_3\cos(\alpha/R)),
\end{gathered} 
\end{equation}
where $\rx=(\rx_1,\rx_2,\rx_3)\in\bS_R$. For $R\in\bN_+$, $\alpha\in (-R,R)$ the map $\cS_{R,\alpha}\,:\,B_R\to \bR^2$ is defined by
\begin{equation}
 \cS_{R,\alpha}(x_1,x_2):= 
 \frac{2(R\sin (\alpha/R) (1-(x_1^2+x_2^2)/4R^2)+x_1 \cos(\alpha/R),x_2)}{1+\cos(\alpha/R)+ (1-\cos(\alpha/R))(x_1^2+x_2^2)/4R^2 -x_1/R\sin(\alpha/R)},
\end{equation}
where $B_R:=\{x\in\bR^2\,|\,|x|< R\}$.
\end{dfn}
\begin{rem}\label{rem:rot_trans}
For all $R\in\bN_+$ and $\alpha\in\bR$ it holds $\cR_{R,\alpha} \jmath_R = \jmath_R \cR_\alpha$. For all $R\in\bN_+$ and $\alpha\in (-R,R)$ it holds $\cT_{R,\alpha} \jmath_R = \jmath_R \cS_{R,\alpha}$ on $B_R$. 
\end{rem}
\begin{dfn}
Let $\alpha\in\bR$ and $R\in\bN_+$. For $f\in C^\infty_\rc(\bR^2)$ we set
\begin{equation}
 \cR_\alpha^*f:=f\circ\cR_\alpha\in C^\infty_\rc(\bR^2),\qquad 
 \cT_\alpha^*f:=f\circ\cT_\alpha\in C^\infty_\rc(\bR^2)
\end{equation}
and for $\phi\in\sD'(\bR^2)$ we set
\begin{equation}
 \cR_\alpha^*\phi:=\phi\circ\cR_{-\alpha}^*\in\sD'(\bR^2),
 \qquad
 \cT_\alpha^*\phi:=\phi\circ\cT_{-\alpha}^*\in\sD'(\bR^2).
\end{equation}
For $f\in C^\infty(\bS_R)$, $\phi\in\sD'(\bS_R)$ we define $\cR_{R,\alpha}^*f,\cT_{R,\alpha}^*f\in C^\infty(\bS_R)$ and $\cR_{R,\alpha}^*\phi,\cT_{R,\alpha}^*\phi\in\sD'(\bS_R)$ by analogous formulas.
% Similarly, for $f\in C^\infty(\bS_R)$ we set
% \begin{equation}
%  \cR_{R,\alpha}^*f:=f\circ\cR_{R,\alpha}\in C^\infty(\bS_R),
%  \qquad
%  \cT_{R,\alpha}^*f:=f\circ\cT_{R,\alpha}\in C^\infty(\bS_R),
% \end{equation} 
% and for $\phi\in\sD'(\bS_R)$ we set 
% \begin{equation}
%  \cR_{R,\alpha}^*\phi:=\phi\circ\cR_{R,-\alpha}^*\in\sD'(\bS_R),
%  \qquad
%  \cT_{R,\alpha}^*\phi:=\phi\circ\cT_{R,-\alpha}^*\in\sD'(\bS_R).
% \end{equation}
\end{dfn}

\begin{dfn}
Let $R\in\bN_+$, $\alpha\in (-R,R)$. For $f\in C^\infty_\rc(\bR^2)$ we set $$\cS_{R,\alpha}^* f:= f\circ \cS_{R,\alpha}\in C^\infty(B_R).$$ For $\phi\in\sD'(\bR^2)$, $\supp\,\phi\subset B_R$, we define $\cS_{R,\alpha}^*\phi\in\sD'(\bR^2)$ by
\begin{equation}
 \langle\cS_{R,\alpha}^*\phi,f\rangle:=
 \langle\phi,\det(\mathrm T\cS_{R,-\alpha})\cS_{R,-\alpha}^*f\rangle 
\end{equation}
for all $f\in C^\infty_\rc(\bR^2)$, where $\det(\mathrm T\cS_{R,-\alpha})$ denotes the Jacobian, i.e. the determinant of the tangent map of $\cS_{R,-\alpha}$.
\end{dfn}

\begin{rem}\label{rem:translation}
For all $\alpha\in\bR$, $a\in\bN_0^2$ and $M\in(0,\infty)$ there exists $C\in(0,\infty)$ such that for all $x\in B_M$ and $R\in(|\alpha|\vee M,\infty)$ it holds 
\begin{equation}
 \cT_{R,\alpha} \jmath_R(x) = \jmath_R \cS_{R,\alpha}(x)
 \qquad
 \textrm{and}
 \qquad
 |\partial^a\cS_{R,\alpha}(x)-\partial^a\cT_\alpha(x)| \leq C/R.
\end{equation}
Noting that $\mathrm T\cT_{-\alpha}=1$ we conclude that for all $\alpha\in\bR$ and $f\in C^\infty_\rc(\bR^2)$ there exists $C$ such that for all sufficiently large $R\in\bN_+$ it holds
\begin{equation}
 \|\det(\mathrm T\cS_{R,-\alpha})\cS_{R,-\alpha}^*f-\cT_{-\alpha}^*f\|_{L^1_2(\bR^2,v_L^{-1/2})}\leq C\,/R.
\end{equation}
\end{rem}
\begin{rem}\label{rem:cylindrical}
Let us note that the algebra of cylindrical functionals $\cF$ separates points in $L^{-1}_2(\bR^2,v_L^{1/2})\subset \sD'(\bR^2)$. Hence, if $\mu_j$, $j=1,2$, are Borel probability measures on $L^{-1}_2(\bR^2,v_L^{1/2})$ such that $\mu_1(F)=\mu_2(F)$ for all $F\in\cF$, then $\mu_1=\mu_2$ by~\cite[Theorem~4.5(a), Ch. 3]{ethier2005}. 
\end{rem}

\begin{prop}\label{prop:euclidean_inv_plane}
Let $\mu$ be a weak limit of a subsequence of the sequence of measures $(\jmath_R^*\sharp\mu_R)_{R\in\bN_+}$ on $\sS'(\bR^2)$. It holds $\mu(F) = \mu(F\circ\cR_\alpha^*)$ and $\mu(F) = \mu(F\circ\cT_\alpha^*)$ for all bounded and measurable $F\,:\,\sS'(\bR^2)\to\bR$ and all $\alpha\in\bR$.
\end{prop}

\begin{proof}
Suppose that the sequence of measures $\jmath_R^*\sharp\mu_R$ on $\sS'(\bR^2)$ converges to $\mu$ along the subsequence $(R_M)_{M\in\bN_+}$. By Remark~\ref{rem:tightness} the measure $\mu$ is concentrated on $L^{-1}_2(\bR^2,v_L^{1/2})$. Hence, by Remark~\ref{rem:cylindrical}, without loss of generality, we can assume that $F\in\cF$ is a cylindrical functional. Note that by the rotational invariance of the measure $\mu_R$ it holds
\begin{equation}
 \mu_R(F_R) = \mu_R(F_R\circ\cR_{R,\alpha}^*),
 \qquad
 \mu_R(F_R) = \mu_R(F_R\circ\cT_{R,\alpha}^*)
\end{equation}
for every $F_R\in\cF_R$ and $\alpha\in\bR$. By Remark~\ref{rem:rot_trans} we have $\cR^*_\alpha\circ\jmath_R^*=\jmath_R^*\circ\cR^*_{R,\alpha}$. Hence, for every $F\in\cF$ we obtain 
\begin{multline}
 \mu(F\circ \cR^*_\alpha-F) =
 \lim_{M\to\infty} 
 \jmath_{R_M}^*\sharp\mu_{R_M}(F\circ \cR^*_\alpha-F) =
 \lim_{M\to\infty} 
 \mu_{R_M}(F\circ \cR^*_\alpha\circ \jmath_{R_M}^* -F\circ\jmath_{R_M}^*)
 \\
 =
 \lim_{M\to\infty} 
 \mu_{R_M}(F\circ\jmath_{R_M}^*\circ \cR_{R_M,\alpha}^*-F\circ\jmath_{R_M}^*) = 
 \lim_{M\to\infty} 
 (\jmath_{R_M}^*\sharp\mu_{R_M}-\mu)(F)
 =
 0.
\end{multline}
This finishes the proof of the rotational invariance.

Let us turn to the proof of the translational invariance. Note that by Remark~\ref{rem:rot_trans} for every $F\in\cF$ and all sufficiently large $R\in\bN_+$ it holds 
\begin{equation}
 \jmath_R^*\sharp\mu_R(F) 
 =
 \mu_R(F\circ\jmath_R^*) 
 =
 \mu_R(F\circ\jmath_R^*\circ\cT_{R,\alpha}^*)
 =
 \mu_R(F\circ\cS_{R,\alpha}^*\circ\jmath_R^*)
 =
 \jmath_R^*\sharp\mu_R(F\circ\cS_{R,\alpha}^*).
\end{equation}
Remark~\ref{rem:translation} implies that for every $\alpha\in\bR$ and $F\in\cF$ there exists $C\in(0,\infty)$ such that for all $\psi\in L^{-1}_2(\bR^2,v_L^{1/2})$ and all sufficiently large $R\in\bN_+$ it holds
\begin{equation}
 |F(\cS_{R,\alpha}^*\psi)-F(\cT_\alpha^*\psi)| \leq (C/R)~\|\psi\|_{L^{-1}_2(\bR^2,v_L^{1/2})}.
\end{equation}
By Proposition~\ref{prop:tightness} and H{\"o}lder's inequality we obtain that
\begin{equation}
 \int_{\sD'(\bS_R)} \|\jmath_R^*\phi\|_{L^{-1}_2(\bR^2,v_L^{1/2})}\,\mu_R(\rd\phi)
\end{equation}
is uniformly bounded in $R\in\bN_+$. Hence, for all $\alpha\in\bR$ and $F\in\cF$ it holds
\begin{equation}
 \lim_{R\to\infty} \jmath_R^*\sharp\mu_R(F\circ\cS_{R,\alpha}^*-F\circ\cT_\alpha^*) =0.
\end{equation}
Consequently,
\begin{multline}
 \mu(F\circ\cT^*_\alpha-F) 
 =
 \lim_{M\to\infty} 
 \jmath_{R_M}^*\sharp\mu_{R_M}(F\circ\cT^*_\alpha-F)
 \\
 =
 \lim_{M\to\infty} 
 \jmath_{R_M}^*\sharp\mu_{R_M}(F\circ\cS_{R_M,\alpha}^*-F)
 =
 \lim_{M\to\infty} 
 (\jmath_{R_M}^*\sharp\mu_{R_M}-\mu)(F)
 =
 0.
\end{multline}
This finishes the proof.
\end{proof}

\appendix

\section{Function spaces}\label{sec:spaces}

\begin{dfn}\label{dfn:weights}
We say that $w\in C^\infty(\bR^d)$ is an admissible weight iff there exist $b\in[0,\infty)$ and $c\in(0,\infty)$ such that $0<w(x)\leq c\,w(y)\,(1+|x-y|)^b$ for all $x,y\in\bR^d$ and for every $a\in\mathbb{N}_0^d$ there exists $c_a\in(0,\infty)$ such that $|\partial^a w (x)|\leq c_a\, w(x)$ for all $x\in\bR^d$.
\end{dfn}

\begin{dfn}\label{dfn:sobolev}
Let $w$ be an admissible weight, $p\in[1,\infty]$ and $\alpha\in\bR$, $n\in\bN_0$. By definition $L_p(\bR^d,w)$ is the Banach space with the norm 
\begin{equation}
 \|f\|_{L_p(\bR^d,w)}:=\|w f\|_{L_p(\bR^d)}.
\end{equation}
The weighted Bessel potential space $L^\alpha_p(\bR^d,w)$ is the Banach space with the norm 
\begin{equation}
 \|f\|_{L^\alpha_p(\bR^d,w)} := \|(1-\Delta)^{\alpha/2}f\|_{L_p(\bR^d,w)}.
\end{equation}
We also set $L^\alpha_p(\bR^d)=L^\alpha_p(\bR^d,1)$. The weighted Sobolev space $W^n_p(\bR^d,w)$ is the Banach space with the norm 
\begin{equation}
 \|f\|_{W^n_p(\bR^d,w)} = \textstyle\sum_{\substack{a\in\bN^d,|a|\leq n}} \|\partial^a f\|_{L_p(\bR^d,w)}.
\end{equation}
For $R\in(0,\infty)$ the Bessel potential space $L^\alpha_p(\bS_R)$ on the round sphere $\bS_R\subset\bR^d$  of radius $R$ is the Banach space with the norm 
\begin{equation}
 \|f\|_{L^\alpha_p(\bS_R)} := \|(1-\Delta_R)^{\alpha/2}f\|_{L_p(\bS_R)},
\end{equation}
where $L_p(\bS_R)$ is the $L_p$ space on $\bS_R$ with respect to the canonical measure $\rho_R$ on $\bS_R$.
\end{dfn}

\begin{rem}\label{rem:sobolev}
The following facts are standard: %
Let $w$ be an admissible weight, $p\in[1,\infty)$, $\alpha\in\bR$ and $n\in\bN_0$. The norms $\|\Cdot\|_{L^\alpha_p(\bR^d,w)}$ and $\|w\Cdot\|_{L^\alpha_p(\bR^d)}$ are equivalent. The Sobolev space $W^n_p(\bR^d,w)$ coincides with the Bessel potential space $L^n_p(\bR^d,w)$ with equivalent norms. The Bessel potential space $L^\alpha_p(\bR^d,w)$ coincides with the Triebel-Lizorkin space $F^\alpha_{p,2}(\bR^d,w)$ with equivalent norms. Furthermore, the Bessel potential space $L_p^{\alpha}(\bR^d,w)$ is continuously embedded in the Besov space $B_{p,\infty}^{\alpha}(\bR^d,w)$ and the Besov space $B_{\infty,1}^{\alpha}(\bR^d,w)$ is continuously embedded in the Bessel potential space $L_\infty^{\alpha}(\bR^d,w)$. These facts can be obtained e.g. from ~\cite[Theorem~6.5, Theorem~6.9]{Triebel3} and \cite[Sec.~2.5.7]{Triebel1}. We note that   \cite[Theorem 6.5 (iii)]{Triebel3} is useful to pass from $\alpha=0$ to $\alpha\in \bR$.
\end{rem}

\begin{rem} \label{Hoelder-remark}
For $\alpha\in\bR$ and $p,q\in[1,\infty]$ we have the following generalized H{\"o}lder inequality
\begin{equation}
 |\langle f,g\rangle_{L_2(\bR^d,w^{1/2})}| \leq C\, \|f\|_{L^\alpha_p(\bR^d,w^{1/p})}\,\|g\|_{L^{-\alpha}_q(\bR^d,w^{1/q})},
 \qquad
 \frac{1}{p}+\frac{1}{q}=1,
\end{equation}
where $\langle \Cdot,\Cdot\rangle_{L_2(\bR^d,w^{1/2})}$ is the scalar product in $L_2(\bR^d,w^{1/2})$ and the constant $C\in(0,\infty)$ depends only on the weight $w$. 
\end{rem}

\begin{thm}\label{thm:embedding}
Let $w,v$ be admissible weights and 
\begin{equation}
 -\infty<\alpha_2\leq \alpha_1 <\infty,
 \qquad
 1\leq p_1\leq p_2\leq \infty.
\end{equation}
(A) The embedding $L^{\alpha_1}_{p_1}(\bR^d,w)\to L^{\alpha_2}_{p_2}(\bR^d,v)$ is continuous if
\begin{equation}
 p_2<\infty,
 \qquad
 \alpha_1-d/p_1\geq \alpha_2-d/p_2 
 \quad\textrm{and}\quad
 \sup_{x\in\bR^d} v(x)/w(x) < \infty.
\end{equation}
(B) The embedding $L^{\alpha_1}_{p_1}(\bR^d,w)\to L^{\alpha_2}_{\infty}(\bR^d,v)$ is continuous if
\begin{equation}
 \alpha_1-d/p_1> \alpha_2
 \quad\textrm{and}\quad
 \sup_{x\in\bR^d} v(x)/w(x) < \infty.
\end{equation}
(C) The embedding $L^{\alpha_1}_{p_1}(\bR^d,w)\to L^{\alpha_2}_{p_2}(\bR^d,v)$ is compact if  
\begin{equation}
p_2<\infty,
\qquad
\alpha_1-d/p_1>\alpha_2-d/p_2
\quad\textrm{and}\quad
\lim_{|x|\to\infty} v(x)/w(x) =0.
\end{equation}
\end{thm}
\begin{proof} Parts (A) and (C) follow from  \cite[Sec. 4.2.3, Theorem]{EdmundsTriebel} and the equivalence between $L^\alpha_p(\bR^d,w)$ and $F^\alpha_{p,2}(\bR^d,w)$ mentioned in Remark~\ref{rem:sobolev} above. Part (B) is covered by   \cite[Sec. 4.2.3, Remark]{EdmundsTriebel} and the embeddings stated 
in Remark~\ref{rem:sobolev}.
\end{proof}

\begin{thm}\label{thm:sobolev_multiplication}
Let $w$ be an admissible weight, $\alpha\in[0,\infty)$ and $p,p_1,p_2\in[1,\infty)$ be such that $1/p=1/p_1+1/p_2$. Then there exists $C\in(0,\infty)$ such that for all $f\in L^\alpha_{p_1}(\bR^d,w^{1/p_1})$ and $g\in L^\alpha_{p_2}(\bR^d,w^{1/p_2})$
\begin{equation}
 \|fg\|_{L^\alpha_p(\bR^d,w^{1/p})}
 \leq C\,\|f\|_{L^\alpha_{p_1}(\bR^d,w^{1/p_1})} 
 \,
 \|g\|_{L^\alpha_{p_2}(\bR^d,w^{1/p_2})}.
\end{equation}
\end{thm}
\begin{proof}
The statement follows from the equivalence of the norms $\|\Cdot\|_{L^\alpha_p(\bR^d,w)}$ and $\|w\Cdot\|_{L^\alpha_p(\bR^d)}$, the fractional Leibniz rule~\cite[Ch.~2]{Muscalu} and Theorem A.5 (A).  Alternatively, one can use \cite[Lemma 5]{BM01}.
\end{proof}

\begin{thm}\label{thm:sobolev_interpolation}
Let $w$ be an admissible weight, $p_1,p_2\in[1,\infty)$, $\alpha_1,\alpha_2\in\bR$, $\theta\in(0,1)$ and 
\begin{equation}
 \alpha=\theta\,\alpha_1+(1-\theta)\,\alpha_2,
 \qquad\qquad
 \frac{1}{p} = \frac{\theta}{p_1} + \frac{1-\theta}{p_2}.
\end{equation}
There exists $C\in(0,\infty)$ such that for all $f\in L^{\alpha_1}_{p_1}(\bR^d,w^{1/p_1})\cap L^{\alpha_2}_{p_2}(\bR^d,w^{1/p_2})$ it holds
\begin{equation}
 \|f\|_{L^\alpha_p(\bR^d,w^{1/p})} \leq C\,\|f\|^\theta_{L^{\alpha_1}_{p_1}(\bR^d,w^{1/p_1})}\, \|f\|^{1-\theta}_{L^{\alpha_2}_{p_2}(\bR^d,w^{1/p_2})}\,.
\end{equation}
\end{thm}
\begin{proof} 
The statement is a consequence of the equivalence of the Bessel potential spaces $L^\alpha_p(\bR^d,w)$ with the Triebel-Lizorkin spaces $F^\alpha_{p,2}(\bR^d,w)$, mentioned in Remark~\ref{rem:sobolev}, and the H{\"o}lder inequality, cf. \cite[Sec 3]{BM01}.
\end{proof}

\newcommand{\pnew}{p'}

\begin{lem}\label{lem:Z_Psi_n}
Let $w\in L_1(\bR^2)$ be an admissible weight, $n\in\{3,4,\ldots\}$, $\delta\in(0,\infty)$ and $\kappa\in(0,2/(n-1)(n-2))$. 
Then there exists $C\in(0,\infty)$ and $p\in[1,\infty)$ such that for all \mbox{$m\in\{1,\ldots,n-1\}$} and $\Psi\in L_2^1(\bR^2,w^{1/2}) \cap L_n(\bR^2,w^{1/n})$, $Z\in L^{-\kappa}_p(\bR^2,w^{1/p})$  it holds
\begin{equation}
 |\langle Z,\Psi^m \rangle_{L_2(\bR^2,w^{1/2})}| 
 \leq
 C\,\|Z\|_{L^{-\kappa}_p(\bR^2,w^{1/p})}^p+
 \delta\,\|\vec\nabla\Psi\|_{L_2(\bR^2,w^{1/2})}^2+
 \delta\,\|\Psi\|_{L_n(\bR^2,w^{1/n})}^n+\delta.
\end{equation}
\end{lem}
\begin{proof}
Let $1/r=(1-\kappa)/n+\kappa/2$, $1/q = m/r$ and $1/\pnew=1-1/q$. By H{\"o}lder's inequality
\begin{equation}
 |\langle Z,\Psi^m \rangle_{L_2(\bR^2,w^{1/2})}| \leq C\,\|Z\|_{L^{-\kappa}_{\pnew}(\bR^2,w^{1/\pnew})}\,\|\Psi^m\|_{L^\kappa_q(\bR^2,w^{1/q})},
\end{equation}
for some $C\in(0,\infty)$. Theorem~\ref{thm:sobolev_multiplication} implies that
\begin{equation}
 \|\Psi^m\|_{L^\kappa_q(\bR^d,w^{1/q})} \leq C\, \|\Psi\|_{L^\kappa_{r}(\bR^d,w^{1/r})}^m
\end{equation}
and Theorem~\ref{thm:sobolev_interpolation} implies that
\begin{equation}
 \|\Psi\|_{L^\kappa_r(\bR^d,w^{1/r})} \leq C\,\|\Psi\|_{L^1_2(\bR^d,w^{1/2})}^\kappa\,\|\Psi\|_{L_n(\bR^d,w^{1/n})}^{1-\kappa}
\end{equation}
for some $C\in(0,\infty)$. Combining the above bounds we obtain
\begin{equation}
 |\langle Z,\Psi^m \rangle_{L^2(\bR^d,w^{1/2})}|
 \leq C\,\|Z\|_{L^{-\kappa}_{\pnew}(\bR^d,w^{1/\pnew})}\,
 \|\Psi\|_{L^1_2(\bR^d,w^{1/2})}^{m\kappa}
 \,\|\Psi\|_{L_n(\bR^d,w^{1/n})}^{m(1-\kappa)}
\end{equation}
for some $C\in(0,\infty)$. Hence, by Young's inequality for every $\delta\in(0,\infty)$ there is $C\in(0,\infty)$ such that
\begin{equation}\label{eq:Z_Psi_n_proof}
 |\langle Z,\Psi^m \rangle_{L_2(\bR^2,w^{1/2})}| 
 \leq C\,\|Z\|_{L^{-\kappa}_{\pnew}(\bR^2,w^{1/\pnew})}^{\pnew}\,
 +\delta\,\|\Psi\|_{L^1_2(\bR^2,w^{1/2})}^2
 +\delta\,\|\Psi\|_{L_n(\bR^2,w^{1/n})}^n.
\end{equation}
We observe that by H{\"o}lder's inequality and the assumption $w\in L_1(\bR^2)$ for all $q,r\in[1,\infty)$ such that $q\leq r$ there exists $C\in(0,\infty)$ such that $\|\Cdot\|_{L_q(\bR^2,w^{1/q})}\leq C\,\|\Cdot\|_{L_r(\bR^2,w^{1/r})}$. Hence, the bound~\eqref{eq:Z_Psi_n_proof} implies the statement of the lemma with $1/p = (2-\kappa(n-1)(n-2))/2n$.
\end{proof}

\begin{lem}\label{lem:embedding_sphere}
Let $p\in[2,\infty)$ and $\alpha=1-2/p$. Then there exists $C\in(0,\infty)$ such that $\|f\|_{L_p(\bS_R)}\leq C\,\|f\|_{L_2^\alpha(\bS_R)}$ for all $f\in L_2^\alpha(\bS_R)$ and all $R\in[1,\infty)$.
\end{lem}
\begin{proof} See e.g.~\cite[Theorem~6]{Be01} or~\cite[Theorem~II.2.7(ii)]{Varopoulos}.
\end{proof}

\section{Mathematical preliminaries}\label{sec:preliminaries}

\begin{lem}\label{lem:trace}
Let $\kappa\in(0,\infty)$. There exists $C\in(0,\infty)$ such that for all $R,N\in[1,\infty)$ it holds 
\begin{equation}
 -C\leq \sum_{l=0}^\infty \frac{(2l+1)}{2R^2\,(1+l(l+1)/R^2)\,(1+l(l+1)/(NR)^2)^\kappa} - \log (N+1) \leq C.
\end{equation}
\end{lem}
\begin{proof}
Observe that the expression in the statement of the lemma coincides with
\begin{equation}
 \int_0^\infty \frac{(2\floor{l}+1)\,\rd l}{2R^2\,(1+\floor{l}(\floor{l}+1)/R^2)\,(1+\floor{l}(\floor{l}+1)/(NR)^2)^\kappa} - \int_0^\infty \frac{\rd l}{(1+l)(1+(1+l)/N)}.
\end{equation}
The absolute value of the above expression is bounded by
\begin{equation}
 \int_0^\infty \left|\frac{(2\floor{R l}+1)/R}{2(1+\floor{R l}(\floor{R l}+1)/R^2)\,(1+\floor{R l}(\floor{R l}+1)/(NR)^2)^\kappa} 
 - \frac{1}{(1+l)(1+(1+l)/N)}\right|\rd l.
\end{equation}
Using $0 \leq l-\floor{R l}/R \leq 1$ we show that there exists $\hat C\in(0,\infty)$ such that the above expression is bounded by
\begin{equation}
 \hat C+ \int_0^\infty \left|\frac{1}{(1+l)(1+l^2/N^2)^\kappa} - \frac{1}{(1+l)(1+(1+l)/N)}\right|\rd l
 \leq C.
\end{equation}
This finishes the proof. 
\end{proof}

\begin{dfn}\label{dfn:tightness_weak_convergence}
Let $\mathcal{X}$ be a topological space and let $(\mu_{n})_{n \in \mathbb{N}_+}$ be a sequence of probability measures  defined on $(\mathcal{X},\mathrm{Borel}(\mathcal{X}))$. The sequence $(\mu_{n})_{n \in \mathbb{N}_+}$
is tight iff for every $\epsilon>0$ there exists a compact set $K_\epsilon\subset\mathcal{X}$ such that $\mu_n(K_\epsilon)\geq 1-\epsilon$ for all $n\in\mathbb{N}_+$. The sequence $(\mu_{n})_{n \in \mathbb{N}_+}$ converges weakly if for every bounded $F\in C(\mathcal{X})$ the sequence of real numbers $(\mu_n(F))_{n\in\mathbb{N}_+}$ converges.
\end{dfn}
\begin{thm}[Prokhorov's theorem] 
Let $\mathcal{X}$ be a separable metric space. A~sequence of probability measures $(\mu_{n})_{n\in\mathbb{N}_+}$ on $(\mathcal{X},\mathrm{Borel}(\mathcal{X}))$ is tight  iff  there exists a diverging sequence of natural numbers $(a_n)_{n\in\mathbb{N}_+}$ such that the sequence $(\mu_{a_n})_{n\in\mathbb{N}_+}$ converges weakly.
\end{thm}

\begin{lem}\label{lem:tightness_criterion}
Let $\mathcal{X},\mathcal{Y}$ be separable normed spaces such that $\imath:\mathcal{X}\to\mathcal{Y}$ is a compact embedding and let $(\mu_{n})_{n\in\mathbb{N}_+}$ be a sequence of probability measures on $(\mathcal{X},\mathrm{Borel}(\mathcal{X}))$. Assume that there exists $M\in(0,\infty)$ such that $\int_{\mathcal{X}}\|x\|_{\mathcal{X}}\,\mu_n(\rd x)\leq M$ for all $n\in\mathbb{N}_+$. Then the sequence of measures $(\nu_n)_{n\in\mathbb{N}_+}$ on $(\mathcal{Y},\mathrm{Borel}(\mathcal{Y}))$ defined by
\begin{equation}
\nu_n(A):=\mu_n(\imath^{-1}(A)),
\qquad\quad
n \in \mathbb{N}_+,
\quad
A \in \mathrm{Borel}(\mathcal{Y}),
\end{equation}
is tight.
\end{lem}
\begin{proof}
Let $\epsilon>0$, $L_\epsilon := \{x\in\mathcal{X}\,|\,\|x\|_{\mathcal{X}}\leq M/\epsilon\}$ and $K_\epsilon:=\overline{\imath(L_{\epsilon})}$. Observe that $K_\epsilon\subset\mathcal{Y}$ is compact. It holds
\begin{equation}
 1-\nu_n(K_{\epsilon}) \leq 1-\mu_n(L_\epsilon) = \mu_n(\|x\|_{\mathcal{X}}>M/\epsilon)
 \leq \epsilon/M~\int_{\mathcal{X}}\|x\|_{\mathcal{X}}\,\mu_n(\rd x)\leq\epsilon.
\end{equation}
This finishes the proof.
\end{proof}

\section{Stochastic estimates}\label{sec:stochastic_estimates}

We recall from \cite[Section 1.1.1]{Nualart}  some basic definitions related  to the Wiener chaos. Let $\mfh$ be a real, separable Hilbert space with scalar product $\lan \,\cdot\,,\, \cdot\, \ran_{\mfh}$. We say that a stochastic process $X=\{X(h) \,|\, h\in \mathfrak{h}\}$ defined in a complete probability space $(\Om,\mcF, \mathbb{P})$ is a   Gaussian process on $\mathfrak{h}$ if $X$ is a centered Gaussian family of random variables such that $\mathbb{E}(X(h) X(g))=\lan h, g\ran_\mfh$ for $h,g\in \mfh$. Now let $H_n, n\in \bN_0$, be the Hermite polynomials. We denote by $\hil_n$ the closed linear subspace of $L^2(\Om,\mathbb{P})$
generated by random variables $\{ H_n(X(h)), h\in \mfh, \|h\|_\mfh=1\}$ and call it the Wiener chaos of order $n$. The subspace
$\bigoplus_{\ell=0}^{n} \hil_\ell$ is called the inhomogeneous Wiener chaos of order $n$.

In our case   $\Om=\sD'(\bS_R)$, $\mathcal{F}=\mathrm{Borel}(\Om)$, $\mathbb{P}=\nu_R$ is the Gaussian measure with covariance $G_R$ and $\mfh=L^{-1}_2(\bS_R)$. Observe that $X^{:m:}_{R,N}=c^{m/2}_{R,N} H_m(X_{R,N}/c_{R,N}^{1/2})$. The choice of the counterterm in~\eqref{eq:counterterm} is dictated by the assumptions of Lemma~\ref{Nualart-Lemma-1-1}.

To facilitate the application of Lemmas~\ref{lem:stochastic_X_uniform}, \ref{lem:stochastic_Y} below  in the proof of Proposition~\ref{prop:uv_limit}, we recall  that convergence  in $L_2(\Om,  \mathbb{P})$ implies convergence in probability, and that the latter property is preserved under composition with continuous functions.

\begin{lem}\label{Nualart-Lemma-1-1}
Let $X,Y$ be two random variables with joint Gaussian distribution such that  $\mathbb{E}(X)=\mathbb{E}(Y)=0$ and $\mathbb{E}(X^2)=\mathbb{E}(Y^2)=1$.
Then, for all $n,m$ we have
\begin{equation}
\mathbb{E}(H_n(X)H_m(Y))=\delta_{n,m}  n! (\mathbb{E}(XY))^n.
\end{equation}
\end{lem}
\begin{proof} 
See \cite[Lemma 1.1.1]{Nualart}. 
\end{proof}

\begin{lem}[Nelson's estimate]\label{lem:nelson}
For every random variable $X$ in an inhomogeneous Wiener chaos of order $n\in\mathbb{N}_+$, cf.  \cite{Nualart}, and every $p\in[2,\infty)$ it holds
\begin{equation}
\bE \big[ |X|^{p}\big]^{\frac{1}{p}} \leq \sqrt{n}(p - 1)^{\frac{n}{2}}\, \bE \big[ X^{2}\big]^{\frac{1}{2}},
\qquad
\bE\exp\Bigg(\frac{n|X|^{2/n}}{6\bE \big[ X^{2}\big]^{\frac{1}{n}}}\Bigg) <\infty.
\end{equation}
\end{lem}
\begin{proof}
The first bound follows from the Nelson hypercontractivity of the Ornstein-Uhlenbeck operator (see e.g.~\cite[Theorem 1.4.1]{Nualart}  or \cite{Nelson}). The second bound is an immediate consequence of the first one. 
\end{proof}

\begin{dfn}\label{dfn:kernel}
For an operator $H\,:\,L_2(\bS_R)\to L_2(\bS_R)$ we denote by $H(\Cdot,\Cdot)$ its integral kernel (if it exists) such that $(Hf)(\rx)=\int_{\bS_R} H(\rx,\ry)\,f(\ry)\,\rho_R(\rd\ry)$. Similarly, for an operator $H\,:\,L_2(\bR^2)\to L_2(\bR^2)$ we denote by $H(\Cdot,\Cdot)$ its integral kernel (if it exists) such that $(Hf)(x)=\int_{\bR^2} H(x,y)\,f(y)\,\rd y$. 
\end{dfn}

\begin{lem}\label{lem:K_hat}
There exists $C\in(0,\infty)$ such that for all $R,N\in\bN_+$ it holds
\begin{equation}
 |\hat K_{R,N}| \leq C\,\frac{1}{1-\Delta_R/N^2},
 \qquad
 |1-\hat K_{R,N}|\leq
 C\,\frac{(1-\Delta_R)/N^2}{1-\Delta_R/N^2},
\end{equation} 
where $\hat{K}_{R,N}$ was introduced in Def.~\ref{def:hat-K}.
\end{lem}
\begin{rem}\label{rem:hat_c}
Recall that $K_{R,N}=(1-\Delta_R/N^2)^{-1}$, $G_R=(1-\Delta_R)^{-1}$ and the counterterms $c_{R,N}$, $\hat c_{R,N}$ were introduced in Eq.~\eqref{eq:counterterm} and Def.~\ref{dfn:hat_X_hat_Y}. Note that the operators $G_R,K_{R,N},\hat K_{R,N}$ commute. Using the above lemma we obtain
\begin{equation}
 |\hat K_{R,N}^2-K_{R,N}^2| \leq |\hat K_{R,N}-K_{R,N}|\,|\hat K_{R,N}+K_{R,N}| \leq 2\,C\,(C+1)\,\frac{(1-\Delta_R)/N^2}{(1-\Delta_R/N^2)^2}.
\end{equation}
Consequently, it holds
\begin{multline}
 |\hat c_{R,N}-c_{R,N}|\leq 
 \mathrm{Tr}(|\hat K_{R,N}^2-K_{R,N}^2| G_R)/4\pi R^2
 \leq
 2C(C+1)\,
 \big[\mathrm{Tr}((1-\Delta_R/N^2)^{-1} (1-\Delta_R)^{-1})
 \\
 -
 \mathrm{Tr}((1-\Delta_R/N^2)^{-2} (1-\Delta_R)^{-1})
 +
 \mathrm{Tr}((1-\Delta_R/N^2)^{-2} (1-\Delta_R)^{-1})/N^2\big]/(4\pi R^2).
\end{multline}
By Lemma~\ref{lem:trace} the RHS of the last inequality above is bounded by a constant independent of $R,N\in\bN_+$.
\end{rem}
\begin{proof}
Note that $\hat K_{R,N}=\sum_{l=0}^\infty (2l+1)\mathrm{Tr}(\hat K_{R,N}\mathcal{P}_{R,l})\,\mathcal{P}_{R,l}$, where $\mathcal{P}_{R,l}\,:\,L_2(\bS_R)\to L_2(\bS_R)$ is defined such that $(2l+1)\mathcal{P}_{R,l}$ is the orthogonal projection onto the eigenspace of the operator $-\Delta_R$ corresponding to the eigenvalue $l(l+1)/R^2$. Consequently, by the triangle inequality for the commuting self-adjoint operators it is enough to show that there exists $C\in(0,\infty)$ such that for all $R,N\in\bN_+$ and $l\in\bN_0$ it holds
\begin{equation}\label{eq:trace_L_operator_bounds_proof}
\begin{gathered}
 (1+l(l+1)/R^2N^2)\,|\mathrm{Tr}(\hat K_{R,N}\mathcal{P}_{R,l})| \leq C,
 \\
 |\mathrm{Tr}((1-\hat K_{R,N})\mathcal{P}_{R,l})| \leq C\,(1+l(l+1))/R^2N^2.
\end{gathered} 
\end{equation}
(To obtain the second bound in the statement of the lemma one combines both estimates in~\eqref{eq:trace_L_operator_bounds_proof}.)
Recall that~\cite[Theorem~2.9]{AK12} the integral kernel of $\mathcal{P}_{R,l}$ is given by
$
 \mathcal{P}_{R,l}(\rx,\ry)=P_l(\rx\cdot\ry/R^2)/4\pi R^2,
$
where $P_l$ is the $l$-th Legendre polynomial. Hence, it holds
\begin{equation}
 \mathrm{Tr}(\hat K_{R,N}\mathcal{P}_{R,l}) = 2\pi\int_0^1 P_l(\cos(\theta/RN))\,RN\sin(\theta/RN)\,h(\theta)\,\rd\theta.
\end{equation}
Using the fact that $RN\sin(\theta/RN)\leq\theta$, $|P_l(\cos\vartheta)|\leq 1$ (cf.~\cite[Sec.~2.7.5]{AK12}) and 
\begin{equation}
 l(l+1) P_l(\cos\vartheta) \sin\vartheta
 =
 -\partial_\vartheta^2(\sin\vartheta\, P_l(\cos\vartheta))
 +
 \partial_\vartheta(\cos\vartheta\, P_l(\cos\vartheta))
\end{equation}
(cf.~\cite[Sec.~2.7.2]{AK12}) we show the first of the bounds~\eqref{eq:trace_L_operator_bounds_proof}. Next, using that $2\pi \int \theta h(\theta) \rd\theta=1$, we obtain that
\begin{equation}
 \mathrm{Tr}((1-\hat K_{R,N})\mathcal{P}_{R,l}) = 2\pi\int_0^1 (P_l(\cos(\theta/RN))\,RN\,\sin(\theta/RN)-\theta)\,h(\theta)\,\rd\theta.
\end{equation}
We note the estimates
\begin{equation}
 0\leq 1-P_l(\cos(\vartheta)) \leq l(l+1)\,(1-\cos(\vartheta))/2\leq l(l+1)\,\vartheta^2/4,
 \quad
 0\leq 1 - \sin(\vartheta)/\vartheta \leq \vartheta^2/6,
\end{equation}
where the second inequality follows from 
$$
1-P_l(u)=P_l(1)-P_l(u)=\int^1_u \frac{d}{dv}P_l(v)dv\leq (1-u)\,l(l+1)/2
$$ 
(cf.~\cite[Sec.~2.7.5]{AK12}). This shows the second bound in~\eqref{eq:trace_L_operator_bounds_proof} and finishes the proof.
\end{proof}

\begin{lem}\label{lem:stochastic_X_not_uniform}
For every $N\in\bN_+$ there exists $C\in(0,\infty)$ such that for all $R\in\bN_+$ it holds
\begin{enumerate}
 \item[(A)] $\bE\|X_{R,N}\|_{L^{1}_2(\bS_R)}^2 
 \leq R^2\, C^2$,
 \item[(B)] $\bE\|\hat X_{R,N}\|_{L^{1}_2(\bS_R)}^2 
 \leq R^2\, C^2$. 
 \end{enumerate}
\end{lem}
\begin{proof}
Recall that $X_{R,N}=K_{R,N}X_R$ and $K_{R,N}=(1-\Delta_R/N^2)^{-1}$. Consequently,
\begin{equation}
\bE \|X_{R,N}\|^2_{L^{1}_{2}(\bS_R)} = \bE\|(1-\Delta_R)^{1/2} (1-\Delta_R/N^2)^{-1} X_R\|^2_{L_2(\bS_R)}.
\end{equation}
By Fubini's theorem  and the fact that $\bE X_R(\rx)X_R(\ry)=G_R(\rx,\ry)$, where $G_R=(1-\Delta_R)^{-1}$, we obtain
\begin{equation}
\bE \|X_{R,N}\|^2_{L^{1}_{2}(\bS_R)} 
=
\mathrm{Tr}\big((1-\Delta_R/N^2)^{-2}\big) 
\leq
N^4\,\mathrm{Tr}\big((1-\Delta_R)^{-2}\big) 
= 
\sum_{l=0}^{\infty} \frac{N^4\,(2l+1)}{(1+l(l+1)/R^2)^2}.
\end{equation}
Now, Item~(A) follows from Lemma~\ref{lem:trace}. Thanks to Lemma~\ref{lem:K_hat} the proof of Item~(B) is the same.
\end{proof}

\begin{lem}\label{lem:stochastic_X_uniform}
For every $\kappa\in(0,\infty)$, $\delta\in[0,2]$ there exists $C\in(0,\infty)$ such that for all $R,N\in\bN_+$ it holds
\begin{enumerate}
 \item[(A)] $\bE\|X_R\|_{L^{-\kappa}_2(\bS_R)}^2 
  \leq R^2\,C^2$,
 \item[(B)] $\bE\|X_R- X_{R,N}\|_{L_2^{-\kappa-\delta}(\bS_R)}^2 
 \leq R^2\,C^2\, N^{-2\delta}$, 
 \item[(C)] $\bE\|X_R-\hat X_{R,N}\|_{L_2^{-\kappa-\delta}(\bS_R)}^2  
 \leq R^2\,C^2\, N^{-2\delta}$. 
\end{enumerate}
\end{lem}

\begin{proof}
Item~(A) follows from Item~(B) and Lemma~\ref{lem:stochastic_X_not_uniform} (A) since, clearly, $\|X_{R,N}\|_{L^{-\kappa}_2(\bS_R)} \leq \|X_{R,N}\|_{L^{1}_2(\bS_R)}$.  To prove Item~(B) note that
\begin{multline}
\bE \|X_R - X_{R,N}\|^2_{L^{-\kappa -\delta}_{2}(\bS_R)} 
= 
\mathrm{Tr}\big((1-\Delta_R)^{-1-\kappa-\delta} (1- (1 -\Delta_R/N^2)^{-1})^2 \big)
\\
\leq N^{-2\delta}\,\mathrm{Tr}\big((1-\Delta_R)^{-1-\kappa}\big)
= 
\sum_{l=0}^{\infty} \frac{N^{-2\delta}\,(2l+1)}{(1+l(l+1)/R^2)^{1+\kappa}}.
\end{multline}
Now, Item~(B) follows from Lemma~\ref{lem:trace}. Thanks to Lemma~\ref{lem:K_hat} the proof of Item~(C) is the same as the proof of Item~(B).
\end{proof}

\begin{lem}\label{lem:stochastic_Y}
Let $R\in\bN_+$. There exists a real-valued random variable $Y_R$ and $C\in(0,\infty)$ such that for all $N\in\bN_+$ it holds
\begin{enumerate}
 \item[(A)] $\bE Y_R^2\leq C^2$,
 \item[(B)] $\bE (Y_R-Y_{R,N})^2\leq C^2\,N^{-1/n}$,
 \item[(C)] $\bE (Y_R-\hat Y_{R,N})^2\leq C^2\,N^{-1/n}$,
 \item[(D)] $\bE (\hat Y_{R,N}-\tilde Y_{R,N})^2\leq C^2\,N^{-1}$.
\end{enumerate}
\end{lem}
\begin{rem}
Recall that $n\in 2\bN_+$, $n\geq 4$, is the degree of the polynomial $P$ and the random variables $Y_{R,N}$ and $\hat Y_{R,N},\tilde Y_{R,N}$ were introduced in Def.~\ref{dfn:X_Y} and Def.~\ref{dfn:hat_X_hat_Y}, respectively.
\end{rem}
\begin{proof}
To prove Items~(A) and~(B) it is enough to show that for every $m\in\{1,\ldots,n\}$ there exists $C\in(0,\infty)$ such that for all $N,M\in\bN_+$ it holds
\begin{equation}
 \bE X^{:m:}_{R,N}(1_{\bS_R})(X^{:m:}_{R,N}-X^{:m:}_{R,M})(1_{\bS_R})\leq C^2\,(N\wedge M)^{-1}.
\end{equation}
Let $G_{R,N,M}:=K_{R,N}G_R K_{R,M}$. By Lemma~\ref{Nualart-Lemma-1-1}
\begin{multline}
 \bE X^{:m:}_{R,N}(1_{\bS_R})(X^{:m:}_{R,N}-X^{:m:}_{R,M})(1_{\bS_R}) \\
 =
 m! \int_{\bS_R^2}\!\!(G_{R,N,N}(\rx,\ry)^m-G_{R,N,M}(\rx,\ry)^m)\,\rho_R(\rd\rx)\rho_R(\rd\ry).
\end{multline}
Consequently, using H{\"o}lder's inequality we obtain that for every $m\in\{1,\ldots,n\}$ there exists $C\in(0,\infty)$ such that for all $N,M\in\bN_+$ and $\rx\in\bS_R$ it holds
\begin{multline}
 |\bE X^{:m:}_{R,N}(1_{\bS_R})(X^{:m:}_{R,N}-X^{:m:}_{R,M})(1_{\bS_R})| 
 \\
 \leq C\,\|(G_{R,N,N}-G_{R,N,M})(\Cdot,\Cdot)\|_{L_m(\bS_R^2)}
 (\|G_{R,N,N}(\Cdot,\Cdot)\|_{L_m(\bS_R^2)}^{m-1}
 +
 \|G_{R,N,M}(\Cdot,\Cdot)\|_{L_m(\bS_R^2)}^{m-1})
 \\
 \leq \hat C\,\|(G_{R,N,N}-G_{R,N,M})(\Cdot,\Cdot)\|_{L_n(\bS_R^2)}
 (\|G_{R,N,N}(\Cdot,\Cdot)\|_{L_n(\bS_R^2)}^{m-1}
 +
 \|G_{R,N,M}(\Cdot,\Cdot)\|_{L_n(\bS_R^2)}^{m-1})
 \\=
 \check C\,\|(G_{R,N,N}-G_{R,N,M})(\rx,\Cdot)\|_{L_n(\bS_R)}
 (\|G_{R,N,N}(\rx,\Cdot)\|_{L_n(\bS_R)}^{m-1}
 +
 \|G_{R,N,M}(\rx,\Cdot)\|_{L_n(\bS_R)}^{m-1}),
\end{multline}
where in the last step above we used the fact that $G_{R,N,N}$ is invariant under rotations and $\hat C=(4\pi R^2)^{( 2-2m/n)} C$, $\check C=(4\pi R^2)^{m/n}\hat C$. By the Sobolev embedding stated in Lemma~\ref{lem:embedding_sphere} there exist $\hat C,C\in(0,\infty)$ such that for all $N\in\bN_+$ it holds
\begin{multline}
 \|(G_{R,N,N}-G_{R,N,M})(\rx,\Cdot)\|^2_{L_n(\bS_R)}
 \leq \hat C\, 
 \|(G_{R,N,N}-G_{R,N,M})(\rx,\Cdot)\|^2_{L_2^{(n-2)/n}(\bS_R)}
 \\
 =
 (4\pi R^2)^{-1}\,\hat C\,\big(\mathrm{Tr}\big[G_R^{(n+2)/n} K_{R,N}^2 (K_{R,N}-K_{R,M})^2\big]\big)
 \leq C\,(N\wedge M)^{-2/n}.
\end{multline}
The last estimate above follows from the bound
\begin{multline}
 \mathrm{Tr}\big[G_R^{(n+2)/n} K_{R,N}^2 (K_{R,N}-K_{R,M})^2\big]
 \\
 \leq
 \mathrm{Tr}\big[(1-\Delta_R)^{-(n+2)/n} |(1-\Delta_R)/N^2+(1-\Delta_R)/M^2|^{1/n}\big]
 \\
 \leq
 2\,(N\wedge M)^{-2/n}\,\mathrm{Tr}\big[(1-\Delta_R)^{-(n+1)/n}\big]
\end{multline}
and Lemma~\ref{lem:trace}. By an analogous reasoning we obtain
\begin{multline}\label{eq:prove_Y_bound_G}
 \|G_{R,N,N}(\rx,\Cdot)\|^2_{L_n(\bS_R)}
 \leq \hat C\, 
 \|G_{R,N,N}(\rx,\Cdot)\|^2_{L_2^{(n-2)/n}(\bS_R)}
 \\
 =
 (4\pi R^2)^{-1}\,\hat C\,\big(\mathrm{Tr}\big[G_R^{(n+2)/n} K_{R,N}^4\big]\big)
 \leq C
\end{multline}
for some constants $C,\hat C$ independent of $N$ and $m$. This proves~(A) and~(B). Thanks to Lemma~\ref{lem:K_hat} the above estimates are also valid when $X_{R,N}$ is replaced with $\hat X_{R,N}$ and $G_{R,N,M}$ is replaced with $\hat G_{R,N,M}:=\hat K_{R,N}G_R\hat K_{R,M}$. Hence,~(C) follows. To prove Item~(D) note that for every $m\in\{1,\ldots,n\}$ there exists $C\in(0,\infty)$ such that for all $N\in\bN_+$ and $\rx\in\bS_R$ it holds
\begin{multline}
 \bE \hat X^{:m:}_{R,N}(1_{\bS_R\setminus\bS_{R,N}})\hat X^{:m:}_{R,N}(1_{\bS_R\setminus\bS_{R,N}})
 \leq \|\hat G_{R,N,N}(\Cdot,\Cdot)\|^m_{L_m(\bS_R\setminus\bS_{R,N}\times \bS_R\setminus\bS_{R,N})}
 \\
 \leq 
 \|\hat G_{R,N,N}(\Cdot,\Cdot)\|^m_{L_m(\bS_R\setminus\bS_{R,N}\times \bS_R)}
 \leq C/N\,
 \|\hat G_{R,N,N}(\rx,\Cdot)\|^m_{L_m(\bS_R)},
\end{multline}
where in the last step we used the rotational invariance of $\hat G_{R,N,N}$ and the fact that the volume of $\bS_R\setminus\bS_{R,N}$ is bounded by $C/N$. To conclude the proof of Item~(D) we use an analog of the bound~\eqref{eq:prove_Y_bound_G} with $G_{R,N,N}$ replaced by $\hat G_{R,N,N}$ and H{\"o}lder inequality.
\end{proof}

\begin{lem}\label{lem:stochastic_bound_infinite_volume}
Let $m\in\bN_+$, $p\in[1,\infty)$, $\kappa\in(0,\infty)$ and $L\in[1,\infty)$. There exists $C\in(0,\infty)$ such that for all $R,N\in\bN_+$, $R\geq L$, it holds
\begin{equation}
 \bE\|\jmath_R^* X^{:m:}_{R,N}\|_{L^{-\kappa}_p(\bR^2,v_L^{1/p})}^p \leq C,
 \qquad
\lim_{N\to \infty}\bE\|\jmath_R^* (X_R-X_{R,N})\|_{L^{-\kappa}_p(\bR^2,v_L^{1/p})}^p =0.
\end{equation} 
\end{lem}
\begin{proof}
By Jensen's inequality it suffices to prove the statement for $p\in2\bN_+$. Let $q=(4/\kappa) \vee 4$. There exists $C\in(0,\infty)$ depending on $p$ and $\kappa$ such that for all $R,N\in\bN_+$ it holds
\begin{multline}
 \bE\|\jmath_R^* X^{:m:}_{R,N}\|_{L^{-\kappa}_p(\bR^2,v_L^{1/p})}^p \leq 
 \|v_L w_L^{-1/q}\|_{L_1(\bR^2)}\,\|\bE ((1-\Delta)^{-\kappa/2}\jmath_R^* X^{:m:}_{R,N}(\Cdot))^p\|_{L_\infty(\bR^2,w_L^{1/q})}
 \\
 \leq C\,\|\bE((1-\Delta)^{-\kappa/2}\jmath_R^* X^{:m:}_{R,N}(\Cdot))^2\|_{L_\infty(\bR^2,w_L^{1/q})}^{p/2}\,,
\end{multline}
where the last bound is a consequence of Lemma~\ref{lem:nelson}. Recall that $\bE X_{R,N}\otimes X_{R,N} = G_{R,N}(\Cdot,\Cdot)$, where $G_{R,N}=K_{R,N}G_RK_{R,N}$. By Lemma~\ref{Nualart-Lemma-1-1}
\begin{equation}
 \bE \jmath_R^*X^{:m:}_{R,N}\otimes \jmath_R^*X^{:m:}_{R,N} = m!\, \tilde G_{R,N}^m,
 \qquad 
 \tilde G_{R,N}:=(\jmath_R^*\otimes\jmath_R^*)G_{R,N}(\Cdot,\Cdot)\,.
\end{equation}
Hence, by Fubini's theorem and explicit formula for the kernel in terms of spherical harmonics
\begin{multline}
\bE(1-\Delta)^{-\kappa/2}\jmath_R^* X^{:m:}_{R,N}\otimes\,
(1-\Delta)^{-\kappa/2}\jmath_R^* X^{:m:}_{R,N}
\\
= m!\,
\big((1-\Delta)^{-\kappa/2} \otimes (1-\Delta)^{-\kappa/2}\big)\, \tilde G_{R,N}^m \in C(\bR^2\times\bR^2).
\end{multline}
Since for $F\in C(\bR^2\times\bR^2)$ it holds $\sup_{x\in\bR^2} F(x,x)\leq \sup_{y\in\bR^2}\sup_{x\in\bR^2} F(x,y)$ we obtain
\begin{multline}
\|\bE((1-\Delta)^{-\kappa/2}\jmath_R^* X^{:m:}_{R,N}(\Cdot))^2\|_{L_\infty(\bR^2,w_L^{1/q})}
\\
\leq
m!\,\sup_y w_L^{1/q}(y)\,\|((1-\Delta)^{-\kappa/2} \otimes 1)\big(1 \otimes (1-\Delta)^{-\kappa/2}) \tilde G_{R,N}^m\big)(\Cdot,y)\|_{L_\infty(\bR^2)}
\\
=
m!\,\sup_y w_L^{1/q}(y)\,\|\big( 1 \otimes (1-\Delta)^{-\kappa/2}) \tilde G_{R,N}^m\big)(\Cdot,y)\|_{L^{-\kappa}_\infty(\bR^2)}.
\end{multline}
By Theorem~\ref{thm:embedding}~(B) there exists $C\in(0,\infty)$ such that for all $R,N\in\bN_+$ the above expression is bounded by
\begin{multline}
\sup_{y\in\bR^2} w_L^{1/q}(y)\,\|\big(1\otimes(1-\Delta)^{-\kappa/2}) \tilde G_{R,N}^m\big)(\Cdot,y)\|_{L_\infty(\bR^2)}
\\
=
\sup_{x\in\bR^2}  \|\big(  1\otimes(1-\Delta)^{-\kappa/2}) \tilde G_{R,N}^m\big)(x,\Cdot)\|_{L_\infty(\bR^2,w_L^{1/q})}
\\
=
\sup_{x\in\bR^2} \|\tilde G_{R,N}^m(x,\Cdot)\|_{L_\infty^{-\kappa}(\bR^2,w_L^{1/q})}
\end{multline}
up to a multiplicative constant $C$, which depends on $m$. The first equality above follows from the fact that for \mbox{$F\in C(\bR^2\times\bR^2)$} it holds $\sup_{x\in\bR^2}\sup_{y\in\bR^2} F(x,y)=\sup_{y\in\bR^2}\sup_{x\in\bR^2} F(x,y)$. By Theorem~\ref{thm:embedding}~(B), since $q>2/\kappa$, the above expression is bounded by
\begin{multline}
\sup_{x\in\bR^2} \|\tilde G_{R,N}^m(x,\Cdot)\|_{L_q(\bR^2,w_L^{1/q})}
\leq
\sup_{x\in\bR^2} \|\tilde G_{R,N}^m(x,\Cdot)\|_{L_q(\bR^2,w_R^{1/q})}
=
\sup_{\rx\in\bS_R} \|G_{R,N}(\rx,\Cdot)^m\|_{L_q(\bS_R)}
\\
=
\sup_{\rx\in\bS_R} \|G_{R,N}(\rx,\Cdot)\|_{L_{mq}(\bS_R)}^m
\leq 
C\sup_{\rx\in\bS_R} \|G_{R,N}(\rx,\Cdot)\|_{L^{(mq-2)/mq}_2(\bS_R)}^m
\\
=
C(4\pi R^2)^{-m/2}\,[\mathrm{Tr}(G_{R,N}(1-\Delta_R)^{(mq-2)/mq}G_{R,N})]^{m/2}.
\end{multline}
The first bound above is true because $R\geq L$. The second bound is a consequence of the Sobolev embedding stated in Lemma~\ref{lem:embedding_sphere}, since $q\geq 2/m$. The first of the bounds from the statement of the lemma follows now from Lemma~\ref{lem:trace} applied with $N'=1$ and $\kappa'=2/mq$. To prove the second of the bounds we use exactly the same strategy as above with $m=1$ and the operator $G_{R,N}$ replaced by $(1-K_{R,N})G_R(1-K_{R,N})$.
\end{proof}

\section*{Acknowledgments}

The financial support by the grant `Sonata Bis' 2019/34/E/ST1/00053 of the National Science Centre, Poland, is gratefully acknowledged.

\end{document}